\documentclass[journal,twoside,web]{ieeecolor}
\usepackage{generic}
\usepackage{cite}
\usepackage{textcomp}

\usepackage[usenames,dvipsnames,svgnames,table]{xcolor}
\usepackage{amssymb,latexsym,amsfonts,amsmath,dsfont}
\usepackage{graphicx}
\usepackage{eso-pic}
\usepackage{epsfig}
\usepackage{mathtools,mathrsfs}
\usepackage{tikz}
\usetikzlibrary{automata}
\usetikzlibrary{shapes}
\usetikzlibrary{calc,shapes,arrows}
\usepackage{mdwlist}
\usepackage{refcount}
\usepackage{comment}
\usepackage{ragged2e}
\usepackage{newunicodechar}
\usepackage{booktabs}

\newunicodechar{∞}{$\infty$}

\let\labelindent\IEEElabelindent
\usepackage{enumitem}

\makeatletter
\def\endthebibliography{%
	\def\@noitemerr{\@latex@warning{Empty `thebibliography' environment}}%
	\endlist
}
\makeatother

\usepackage{keyval,trig}

\definecolor{mediumblue}{rgb}{0.0, 0.0, 0.8}
\definecolor{mediumcandyapplered}{rgb}{0.89, 0.02, 0.17}
\definecolor{nazar}{rgb}{0.7, 0.5, 0.9}
\makeatletter
\let\NAT@parse\undefined
\makeatother
\usepackage[colorlinks=true, citecolor=mediumcandyapplered, linkcolor=mediumblue, urlcolor = nazar, final]{hyperref}

\definecolor{lightblue}{rgb}{0.30,0.75,0.93}
\newcommand{\greensquare}{\tikz\fill[green!20!white] (0,0) rectangle (2mm,2mm);}
\newcommand{\redsquare}{\tikz\fill[red!20!white] (0,0) rectangle (2mm,2mm);}

\DeclareRobustCommand\sampleline[1]{%
	\tikz\draw[#1] (0,0) (0,\the\dimexpr\fontdimen22\textfont2\relax)
	-- (2em,\the\dimexpr\fontdimen22\textfont2\relax);%
}

\usepackage{algorithm}
\usepackage{algorithmic}

\newtheorem{theorem}{Theorem}
\newtheorem{lemma}{Lemma}

\newtheorem{problem}{Problem}

\newtheorem{definition}{Definition}

\newtheorem{remark}{Remark}

\usepackage[many]{tcolorbox}
\usetikzlibrary{calc}
\tcbuselibrary{skins}

\newtcolorbox{resp}[1][]{%
	enhanced jigsaw,%
	colback=gray!2!white,%
	colframe=gray!80!black,%
	size=small,%
	boxrule=1pt,%
	halign title=flush center,%
	coltitle=black,%
	breakable,%
	drop shadow=black!50!white,%
	attach boxed title to top left={xshift=1cm,yshift=-\tcboxedtitleheight/2,yshifttext=-\tcboxedtitleheight/2},%
	minipage boxed title=3cm,%
	boxed title style={%
		colback=white,%
		size=fbox,%
		boxrule=1pt,%
		boxsep=2pt,%
		underlay={%
			\coordinate (dotA) at ($(interior.west) + (-0.5pt,0)$);
			\coordinate (dotB) at ($(interior.east) + (0.5pt,0)$);
			\begin{scope}[gray!80!black]
				\fill (dotA) circle (2pt);
				\fill (dotB) circle (2pt);
			\end{scope}
		}%
	},%
	#1%
}

\xdefinecolor{MATLABblue}{rgb}{0.0,0.45,.74}
\xdefinecolor{MATLABred}{rgb}{0.85,0.33,.1}

\newcommand{\R}{{\mathbb{R}}}

\makeatletter
\def\@opargbegintheorem#1#2#3{\textit{#1\ #2} \textit{(#3):}}
\makeatother

\newcommand\Xright{\overrightarrow{\mathbb{X}}}
\newcommand\Bxright{\overrightarrow{\mathcal{B}}(x)}

\newcommand\thetaTildeI{\left[\begin{array}{c} \mathds{I} \\\Omega^\top \end{array}\right]}

\def\BibTeX{{\rm B\kern-.05em{\sc i\kern-.025em b}\kern-.08em
		T\kern-.1667em\lower.7ex\hbox{E}\kern-.125emX}}
	
	\usepackage{etoolbox}
	
	\makeatletter
	\AtBeginDocument{%
		\pagestyle{headings}
		
		\patchcmd{\@oddhead}{\\[-19pt]}{\\[-8pt]}{}{}%
		\patchcmd{\@evenhead}{\\[-19pt]}{\\[-8pt]}{}{}%
	}
	\makeatother
	
\begin{document}
	
\title{{Data-Efficient Control of Polynomial Systems via Physics-Guided Quadratic Constraints}}
 \author{MohammadHossein Ashoori, \IEEEmembership{Student Member,~IEEE}, Ali Aminzadeh, \IEEEmembership{Member,~IEEE}
 	\\ Amy Nejati, \IEEEmembership{Senior Member,~IEEE}, and Abolfazl Lavaei, \IEEEmembership{Senior Member,~IEEE}
 	\thanks{M.H. Ashoori, A. Nejati, and A. Lavaei are with the School of Computing, Newcastle University, NE4 5TG Newcastle Upon Tyne, United Kingdom (e-mails: {\tt\small{\{m.ashoori2,amy.nejati,abolfazl.lavaei\}@newcast\\le.ac.uk}}).}
 	 \thanks{A. Aminzadeh is with the Automation Technology and Mechanical Engineering Unit at the Faculty of Engineering and Natural Sciences, Tampere University, Finland (e-mail: {\tt\small{ali.aminzadeh@tuni.fi}}).}
}

\maketitle
\begin{abstract}
This work addresses the critical challenge of guaranteeing safety for complex dynamical systems where precise mathematical models are uncertain and data measurements are corrupted by noise. We develop a {\emph{physics-guided}}, direct data-driven framework for synthesizing \emph{robust} safety controllers for discrete-time nonlinear polynomial systems that are subject to \emph{unknown-but-bounded} disturbances. To do so, we introduce a notion of safety through \emph{robust} control barrier certificates, which ensure avoidance of unsafe regions, offering a less conservative alternative to existing methods based on robust invariant sets.  {To achieve data efficiency, we further integrate physical information, formulated as \emph{quadratic constraints} on system and control matrices,} with observed noisy data. This integration drastically reduces data requirements, enabling robust safety analysis with significantly \emph{shorter} trajectories compared to purely data-driven methods. The proposed synthesis procedure is formulated as a sum-of-squares optimization program that systematically designs the {barrier} and its associated {controller} by leveraging both collected data and underlying physical laws. The efficacy of our framework is demonstrated on three benchmark systems, confirming its ability to offer robust safety guarantees with \emph{reduced} data demands.
\end{abstract}


\section{Introduction}\label{Sec: Introduction}

\IEEEPARstart{S}{afety}-critical systems are embedded in nearly every aspect of modern life, ranging from self-driving vehicles and air traffic control systems to healthcare technologies. Failures in these systems can result in severe consequences, such as loss of life, environmental damage, or substantial financial losses~\cite{mcgregor2017analysis}. As these systems become increasingly complex, ensuring their safe operation necessitates control approaches that rigorously enforce safety constraints, preventing the system from entering unsafe regions despite internal complexity or external disturbances.

{\bf State of the Art.} In recent years, there has been growing interest in ensuring the safety of dynamical systems through \emph{control barrier certificates (CBCs)}, initially introduced in \cite{prajna2004safety,wieland2007constructive}. By imposing specific inequality constraints on a candidate barrier function and its derivative (or difference) along the system's dynamics, analogous to Lyapunov functions, CBCs ensure that trajectories remain within predefined safe regions. Accordingly,  the existence of a CBC provides a formal (probabilistic) certificate of safety. While CBCs have been widely used for formal safety verification and controller synthesis in complex systems, both deterministic \cite{ames2019control,santoyo2021barrier,xiao2023safe,wooding2024protect} and stochastic \cite{zhang2012safety,clark2019control,NEURIPS2020_barrier,lavaei2024scalable,nejati2024context,lavaei2022automated}, they typically rely on the availability of a \emph{precise} mathematical model of the system. This assumption is often restrictive, as real-world systems are invariably affected by parameter uncertainties, unmodeled dynamics, or external disturbances.

{While classical robust control methods (\emph{e.g.,} $\mu$-synthesis~\cite{zhou1996robust}) have widely been employed to handle parameter uncertainties by designing controllers resilient to predefined uncertainty sets, increasing attention has recently been devoted to \emph{data-driven} methods, which adapt established robust control tools to characterize uncertainty sets directly from collected data and available physical prior information, rather than relying solely on static \emph{a priori} model assumptions.}  Data-driven methods are generally classified into \emph{indirect} and \emph{direct} approaches \cite{wang2016indirect,dorfler2022bridging}. More precisely, \emph{indirect} methods follow a two-step procedure: first, they identify a system model from collected data, and second, they apply traditional \emph{model-based} control design to this identified model \cite{wang2016indirect}. This process can be fragile, as errors from the system identification stage can propagate and compromise the final safety guarantees. Furthermore, it can be computationally intensive, especially for complex systems where identifying a high-fidelity model is a challenge in itself. In contrast, \emph{direct} data-driven approaches offer a more streamlined alternative by bypassing explicit model construction and designing controllers directly from system measurements \cite{krishnan2021direct,de2019formulas,martin2023guarantees}. This approach mitigates the \emph{two-level} computational costs of indirect approaches and eliminates errors from model approximation, making them especially valuable for complex systems where detailed modeling is impractical or computational resources are constrained.  Nevertheless, ensuring safety solely from observed data without access to the underlying model remains a critical challenge.

In the realm of direct data-driven methods, the \emph{scenario approach}~\cite{calafiore2006scenario,campi2009scenario} has emerged as a widely used framework for robust control design. This methodology constructs solutions from data and subsequently relates them back to the original system via intermediate formulations that encode chance constraints~\cite{esfahani2014performance,margellos2014road}. Despite its strengths, the scenario approach relies on a key assumption: the data must consist of \emph{independent and identically distributed (i.i.d.)} samples. In practice, this implies that each sample should originate from a distinct, independent input-output trajectory~\cite{calafiore2006scenario}, often necessitating access to a \emph{large number} of independent trajectories. As a result, its applicability is mainly suited to \emph{simulator-based} environments, where such independent data collection is feasible.

  An alternative to the scenario approach in direct data-driven methods is the \emph{non-i.i.d.} trajectory-based framework. Instead of relying on multiple independent samples, this method relies on information from a \emph{single} input-state trajectory observed over a \emph{finite-horizon} experiment to perform control analysis~\cite{de2019formulas,van2020noisy,luppi2024data,monshizadeh2024meta,samari2024single}. Building on the concept of \emph{persistent excitation}, this approach requires the trajectory to satisfy a rank condition for specific system classes to sufficiently capture the system's behavior, {as established by Willems \textit{et al.'s} fundamental lemma \cite{willems2005note} and generalized to multiple trajectories in~\cite{van2020willems}. From the perspective of the behavioral framework, when a trajectory is persistently excited, it ensures that the range space of the data matrix spans the full behavior of the system. Consequently, the acquired data itself serves as the system model, facilitating direct control synthesis while bypassing the intermediate step of explicit parametric identification.} For the sake of fairness, it should be noted that while the scenario approach typically requires a large amount of i.i.d. data, it can handle a general class of nonlinear systems. In contrast, trajectory-based methods relying on persistence of excitation are currently applicable only to certain classes of nonlinear systems, such as those with polynomial dynamics.

While single-trajectory approaches are especially beneficial in settings where collecting multiple independent trajectories is impractical or infeasible, such methods typically rely on \emph{long-horizon} trajectories for complex systems to offer control analysis and design. This raises a key question: In the absence of an exact mathematical model, how can we leverage fundamental physical principles of the system’s dynamics to conduct formal safety analysis using substantially \emph{shorter trajectories}?

{\bf Central Contribution.} Motivated by this pivotal question, this paper develops {a direct data-driven framework for synthesizing robust control barrier certificates (R-CBCs) and their corresponding robust safety controllers (R-SCs) for nonlinear polynomial systems subject to unknown-but-bounded disturbances. Our primary contributions are as follows:
\begin{enumerate}[label=(\roman*)]
	\item We introduce a less conservative notion of safety through the proposed R-CBCs. In contrast to approaches based on robust control invariant (RCI) sets (\emph{e.g.,}~\cite{niknejad2023physics,luppi2024data}), which aim to render the entire safe set invariant, our framework only requires trajectories to originate from an initial set that is a subset of the safe set, while allowing them to evolve freely thereafter provided they never cross the safety threshold defined by the unsafe level set, thereby reducing conservativeness.
    \item To achieve data efficiency, we augment the purely data-driven synthesis with a \emph{physics-guided quadratic constraint}. This constraint utilizes approximate nominal parameter estimates, derived from first principles, alongside predefined uncertainty bounds. By integrating this physical prior, conceptually serving as side information \cite{ahmadi2023learning}, with noisy data from a single, short finite-horizon trajectory, we formulate an \emph{end-to-end} sum-of-squares (SOS) optimization program, drastically mitigating data requirements.
\end{enumerate}}

{\bf Existing Relevant Literature.} Foundational data-driven algebraic tools, such as Petersen's lemma \cite{petersen1987stabilization} and the S-procedure \cite{9308978}, have been widely utilized to handle noisy data \cite{bisoffi2022data,bisoffi2021trade,guo2024data}. While these studies primarily focus on asymptotic stabilization via intermediate geometric proxy objectives (\emph{e.g.,} explicitly minimizing uncertainty bounds), our synthesis framework provides formal safety guarantees via a direct, end-to-end SOS program.

Closely related to our work is \cite{luppi2024data}, which addresses data-driven safety for polynomial systems using RCI sets. Our R-CBC framework reduces the geometric conservatism of \cite{luppi2024data} by allowing the initial region to be a subset of the safe set, rather than requiring strict invariance of the entire safe set. It also enables direct synthesis via a single-shot SOS optimization, which avoids the potentially conservative overapproximation of system and control matrix sets. Moreover, it reduces the data requirement by incorporating prior physical knowledge of the system. It is worth emphasizing that the approach in~\cite{luppi2024data} renders more general safe sets, defined by polynomial inequalities, robustly invariant, whereas our framework restricts the R-CBC to a quadratic form. Moreover, the safety guarantees in~\cite{luppi2024data} for continuous-time systems are established over an infinite time horizon, while the continuous-time counterpart of our approach provides safety guarantees over a finite time horizon.

Regarding the incorporation of inaccurate physical knowledge, our physics-guided quadratic constraint is conceptually similar to the quadratic constraint relating the target and source systems in \cite{li2023data}, which was employed for transfer stabilization. In addition to adapting this mathematical formulation to improve data efficiency in safety certification, we further extend the framework to nonlinear polynomial systems.

Finally, our framework addresses the practical limitations of prior data-driven verification and synthesis methods. While scenario optimization techniques \cite{nejati2023formal} and their physics-guided extensions \cite{aminzadeh2024physics} provide safety verification, they rely on data-heavy i.i.d. sampling (often requiring thousands of trajectories) and do not perform controller synthesis. Conversely, while early trajectory-based SOS methods \cite{samari2024single} synthesize controllers from a single trajectory, they assume idealized settings without unknown disturbances, noisy data, or physical prior information for improving data efficiency. Furthermore, existing trajectory-based physics-guided approaches \cite{niknejad2023physics} remain strictly limited to linear systems. Other model-based robust approaches \cite{takano2018application, xu2015robustness, cosner2023robust, jankovic2018robust, kang2023verification} or imitation learning methods \cite{lindemann2024learning} similarly contrast with our data-efficient formulation.

{\bf Organization.} The rest of the paper is structured as follows. Section \ref{Discrete_Problem_Description} is dedicated to describe discrete-time nonlinear polynomial systems, including mathematical notations, formal definitions of the system and the corresponding R-CBC. Building on this foundation, Section \ref{Physics-Guided} presents our physics-guided data-driven framework, which is designed to handle systems with unknown-but-bounded disturbances. Within this framework, Section \ref{R-CBC_Design} details the core contribution of our work: a systematic method for jointly synthesizing the R-CBC and its corresponding R-SC by leveraging both collected data and underlying physical laws. To verify the efficacy of our work, Section \ref{Sec_Case_studies} provides three nonlinear polynomial case studies, and finally, Section \ref{Sec_Conclusion} concludes the paper.

\section{Problem Description}\label{Discrete_Problem_Description}

\subsection{Notation}

We denote the set of real numbers by $\mathbb{R}$, while $\mathbb{R}_{\geq 0}$ and $\mathbb{R}_{>0}$ represent the sets of non-negative and positive real numbers, respectively. The sets of non-negative and positive integers are denoted by $\mathbb{N} = \{0,1,2,\dots\}$ and $\mathbb{N}^+ = \{1,2,\dots\}$, respectively. The notation $\mathbb{R}^n$ represents an $n$-dimensional Euclidean space, whereas $\mathbb{R}^{n \times m}$ denotes the space of real matrices with $n$ rows and $m$ columns.  A vector with the components $x_i \in \mathbb{R}$ is denoted by $x = [x_1;  \dots ;x_N]$. Given $N$ vectors $x_i \in \R^n$, the corresponding matrix comprising these vectors is expressed as $x=[x_1 \, \, \ldots \,\, x_N] \in \R^{n \times N}$. We denote by $ \Vert \cdot \Vert_2 $ the \emph{spectral} norm of a matrix or the \emph{Euclidean} norm of a vector. The identity matrix of size $n \times n$ is expressed as $\mathds{I}_n$, while $\mathds{I}$ denotes an identity matrix with an appropriate dimension. Additionally, the symbol $\boldsymbol{0}_n \in \mathbb{R}^n$  represents a vector with all zero components.
A \emph{symmetric} matrix \( \mathcal{A} \) is denoted positive definite by \( \mathcal{A} \succ 0 \) and positive semi-definite by \( \mathcal{A} \succeq 0 \). The transpose of \( \mathcal{A} \) is represented as \( \mathcal{A}^\top \). Additionally, in a \emph{symmetric} matrix, $*$ represents the transposed entry corresponding to its symmetric counterpart. The minimum and maximum eigenvalues of a square matrix $\mathcal{A}$ are denoted by $\lambda_{\text{min}}(\mathcal{A})$ and $\lambda_{\text{max}}(\mathcal{A})$, respectively. 

\subsection{Discrete-Time Nonlinear Polynomial Systems}

We begin with discrete-time input-affine nonlinear polynomial systems, as formalized in the following definition.

\begin{definition}[\textbf{dt-IANPS}]\label{Def_1}
	A discrete-time input-affine nonlinear polynomial system (dt-IANPS) is defined as:
	\begin{equation}\label{Eq_BBox}
		\Sigma\!: x(k+1) = f(x(k)) + g(x(k))u(k) + \omega(k),  \quad k \in \mathbb{N},
	\end{equation}
	where \(x \in X\) represents the state, \(u \in U\) is the control input, and \(\omega \in W\) is the \emph{unknown-but-bounded} disturbance. The {known} sets \(X, W \subseteq \mathbb{R}^{n}\), and \(U \subseteq \mathbb{R}^{l}\) correspond to the state, disturbance, and input sets, respectively. {Moreover, \(f:X\to X\), with \(f(\boldsymbol{0}_n)=\boldsymbol{0}_n\), and \(g:X\to\mathbb{R}^{n\times l}\)
are polynomial maps.}
\end{definition}
The dynamics of (\ref{Eq_BBox}) can be expressed equivalently as
\begin{equation}\label{Eq_Llike}
	\Sigma\!:  x(k+1) = A \mathcal{M}(x(k)) + B{\mathcal{Q}}(x(k))u(k) + \omega(k),
\end{equation}
where $A \in \mathbb{R}^{n\times m}, 
B \in \mathbb{R}^{n\times q}$ are system and control matrices, while $\mathcal{M}(x)\in \mathbb{R}^{m}$, with $\mathcal{M}(\boldsymbol{0}_n) = \boldsymbol{0}_m$, and ${\mathcal{Q}}(x) \in \mathbb{R}^{q\times l}$ are a vector and a matrix of monomials in the components of the state vector $x$, respectively. 

We denote by $x_{x_0 u \omega}(k)$ the \emph{state trajectory} of $\Sigma$ at time $k \in \mathbb{N}$, under the input and disturbance signals $u(\cdot)$ and $\omega(\cdot)$, starting from an initial condition $x_0 = x(0)$.

In this work, both matrices $A$ and $B$ are considered \emph{unknown}, while an \emph{extended dictionary}~\cite{de2023data} (\emph{i.e.,} a library or family of functions) for $\mathcal{M}(x)$ and $\mathcal{Q}(x)$ is assumed to be available, encompassing a sufficiently rich set of terms to represent the true system dynamics, albeit with the inclusion of some superfluous terms. Since $\mathcal{M}(\boldsymbol{0}_n) = \boldsymbol{0}_m$, without loss of generality, one can find a polynomial matrix $\mathcal{C}(x) \in \mathbb{R}^{m \times n}$, where
\begin{align}\label{transform}
	\mathcal{M}(x) = \mathcal{C}(x) x.
\end{align}
This transformation facilitates expressing our conditions in terms of $\mathcal{C}(x)$ (cf. \eqref{Eq_CBC3_theorem}), simplifying the computational complexity. We also assume the disturbance $\omega$ is unknown but bounded.

 \begin{remark} [\textbf{On \(f(\boldsymbol{0}_n)=\boldsymbol{0}_n\)}]\label{Remark: origin_equilibrium}
   In a data-driven setting where the exact system dynamics are unknown, the structural assumption $f(\boldsymbol{0}_n) = \boldsymbol{0}_n$  can be empirically assessed prior to controller synthesis. Specifically, by initializing the system at the origin and applying zero control input, one can examine whether the resulting state trajectory remains confined to a small neighborhood around the origin, exhibiting only bounded fluctuations induced by the disturbance $\omega$. Such behavior provides practical evidence supporting the assumption that $f(\boldsymbol{0}_n) = \boldsymbol{0}_n$. This assumption implies that the origin is a natural equilibrium point for the underlying system in the absence of control inputs and external disturbances.
\end{remark}

\begin{remark}[\textbf{Dictionary for $\mathcal{M}(x)$ and $\mathcal{Q}(x)$}]\label{Remark: dictionary}
	Employing a \emph{rich} dictionary for $\mathcal{M}(x)$ and $\mathcal{Q}(x)$ is typically not a limiting factor. In numerous real-world applications, especially in electrical and mechanical engineering, the governing dynamics are frequently determined by underlying physical laws and first-principle modeling (cf. Section~\ref{Subsec_PIex}). While such laws dictate the structural expressions of $\mathcal{M}(x)$ and $\mathcal{Q}(x)$, the exact values of system parameters are often unknown (\emph{i.e.,} $A$ and $B$), which is consistent with our assumption that $A$ and $B$ are entirely unavailable. In addition, knowing the maximum degree of $\mathcal{M}(x)$ allows one to enumerate all admissible monomial combinations up to that degree (see benchmark case studies). 
\end{remark}

As this work focuses on developing a robust safety certificate for the unknown dt-IANPS in~\eqref{Eq_Llike}, the next subsection provides a formal definition of robust CBCs.

\subsection{Robust Control Barrier Certificates}

\begin{definition}[\textbf{R-CBC}] \label{R-CBC def}
	Consider a dt-IANPS $\Sigma$, with $X_{\mathbf i},X_{\mathbf u} \subseteq X $ being its \emph{initial} and \emph{unsafe} sets, respectively. A function $\mathcal{B}: X \to \mathbb{R}_{\geq 0}$ is called a robust control barrier certificate (R-CBC) for $\Sigma$ over a time horizon $[0,\mathcal T)$ if there exist $\gamma_{\mathbf i}, \gamma_{\mathbf u}  \in \mathbb{R}_{>0}$, with $\gamma_\mathbf i < \gamma_\mathbf u$, $\delta \in \mathbb{R}_{> 0},$ and $\lambda \in (0,1)$, such that
	\begin{subequations}\label{Eq_CBC}
		\begin{align}
			\mathcal{B}(x) &\leq \gamma_{\mathbf i}, & \forall x &\in X_{\mathbf i}, \label{Eq_CBC1} \\
			\mathcal{B}(x) &\geq \gamma_{\mathbf u}, & \forall x &\in X_{\mathbf u}, \label{Eq_CBC2} \end{align}
		and ${\forall x \in X, \;}\exists u \in U $ such that $\forall \omega \in W$,
		\begin{align}
			\mathcal{B}\big(A \mathcal{M}(x) + B{\mathcal{Q}}(x)u + \omega\big) \leq \lambda\mathcal{B}(x) + \delta,\label{Eq_CBC3}
		\end{align}
	\end{subequations}
	with (potentially small) $\delta$ satisfying
	\begin{align}\label{New10}
		\delta <  (\gamma_\mathbf u -\lambda^\mathcal{T}\gamma_\mathbf i ) \dfrac{1-\lambda}{1-\lambda^\mathcal{T}}.
	\end{align}
\end{definition}

As shown in Definition~\ref{R-CBC def}, the R-CBC imposes two conditions on the initial and unsafe level sets of the barrier function (\emph{i.e.,} \eqref{Eq_CBC1}–\eqref{Eq_CBC2}), and one condition along the system dynamics (\emph{i.e.,} \eqref{Eq_CBC3}). If there exists a level set of the barrier function that successfully isolates the unsafe set from all possible trajectories starting within the specified initial set, then this function serves as a certificate of the system's safety. It is important to note that the parameter $\delta$ in \eqref{Eq_CBC3} quantifies the level of robustness with respect to the unknown-but-bounded disturbance $\omega$ (cf.~\eqref{Eq_noise_limit}).

To illustrate the efficacy of the R-CBC in guaranteeing the robust safety of dt-IANPS in both \emph{infinite and finite} time horizons, we introduce the following theorem as the first contribution of our work.

\begin{theorem}[\textbf{Safety Guarantee for dt-IANPS}]\label{Lemma imp}
	Consider a dt-IANPS with an R-CBC $\mathcal{B}$ and the control input $u$ that jointly satisfy conditions \eqref{Eq_CBC}.
	\begin{itemize}
		\item If 
		\begin{subequations}
			\begin{align}\label{New9}
				\delta \leq \gamma_\mathbf u(1 - \lambda),
			\end{align}
			then for any initial state $x_0 \in X_\mathbf i$, all system trajectories remain outside the unsafe region $X_\mathbf u$ for all time (\emph{i.e.,} an \emph{infinite} time horizon), implying $x_{x_0 u \omega}(k) \notin X_\mathbf u$.
			\item If 
			\begin{align}\label{new8}
				\gamma_\mathbf u(1 - \lambda) < \delta \leq (\gamma_\mathbf u -\lambda^\mathcal{T}\gamma_\mathbf i) \dfrac{1-\lambda}{1-\lambda^\mathcal{T}},
			\end{align}  
		\end{subequations}
		then all system trajectories avoid $X_\mathbf u$ within the \emph{finite} time horizon $\mathcal{T}$. 
	\end{itemize}
\end{theorem}

\begin{proof}
		The proof consists of two parts:
		\begin{itemize}
			\item First, we analyze the case where $\delta \leq \gamma_\mathbf u(1 - \lambda)$, which provides \emph{infinite} time horizon guarantees. According to \eqref{Eq_CBC1}, the initial state satisfies $\mathcal{B}(x(0))\leq \gamma_\mathbf i <\gamma_\mathbf u$.  We now  show that if $\mathcal{B}(x(k))< \gamma_\mathbf u$, then $\mathcal{B}(x(k+1))< \gamma_\mathbf u$. Given that $\delta \leq \gamma_\mathbf u(1 - \lambda) $ as per \eqref{New9}, and in accordance with \eqref{Eq_CBC3}, one has $\mathcal{B}(x(k+1))< \lambda\gamma_\mathbf u + \gamma_\mathbf u(1 - \lambda) = \gamma_\mathbf u$. Hence, according to \eqref{Eq_CBC2}, it follows that $x(k+1) \notin X_\mathbf u$. Therefore, all system trajectories will remain outside $X_\mathbf u$ within \emph{infinite} time horizons. It is also clear that if $\mathcal{T} \to \infty$ in \eqref{New10}, the condition in \eqref{New9} is recovered. Furthermore, since $\gamma_\mathbf u > 0$ and $\lambda \in (0,1)$, the inequality in \eqref{New9} always holds true for $\delta=0$, which implies the infinite-horizon guarantee.
			\item Now let us consider the case where $\delta > \gamma_\mathbf u(1 - \lambda) $, which offers safety guarantees for \emph{finite} time horizons. Starting from an initial state $x(0) \in X_\mathbf i$, according to~\eqref{Eq_CBC3}, after $\mathcal{T}$ time steps, one has\vspace{0.2cm}
			\begin{align*}
					\mathcal{B}(x(\mathcal{T}))&\leq\lambda\mathcal{B}(x(\mathcal{T}-1)) + \delta\\
					&\leq\lambda(\lambda\mathcal{B}(x(\mathcal{T}-2)) + \delta) + \delta
			\end{align*}
			\begin{align*}
				&~\vdots\\
				&\leq\lambda^{\mathcal{T}}\mathcal{B}(x(0)) + \delta(1+\dots+\lambda^{\mathcal{T}-1}))\\
				&=\lambda^{\mathcal{T}}\mathcal{B}(x(0)) + \delta \dfrac{1-\lambda^\mathcal{\mathcal{T}}}{1-\lambda}\\
				&\stackrel{\eqref{Eq_CBC1}}{\leq} \lambda^{\mathcal{T}}\gamma_\mathbf i + \delta \dfrac{1-\lambda^\mathcal{T}}{1-\lambda}\\			  &\stackrel{\eqref{new8}}{<}\lambda^{\mathcal{T}}\gamma_\mathbf i + (\gamma_\mathbf u -\lambda^\mathcal{T}\gamma_\mathbf i) (\dfrac{1-\lambda}{1-\lambda^\mathcal{T}})(\dfrac{1-\lambda^\mathcal{T}}{1-\lambda})\\
				&= \gamma_\mathbf u,
			\end{align*}
			implying $x(\mathcal{T}) \notin X_\mathbf u.$ Hence, one can conclude that all system trajectories will remain outside $X_\mathbf u$ within the \emph{finite} time horizon $\mathcal{T}$ satisfying the upper bound in \eqref{new8}, which completes the proof.
		\end{itemize}
\end{proof}

\begin{remark}[\textbf{Safety Guarantee Horizon \& R-CBC Novelty}]
~In traditional model-based setups, standard CBC conditions cannot ensure \emph{infinite-time} safety under persistent worst-case disturbances. Incorporating the decay rate $\lambda$ in \eqref{Eq_CBC3} resolves this limitation by establishing a mathematical bound for infinite-time safety, while also yielding a corresponding closed-form bound for finite-time safety guarantees. The achievable time horizon, as established in Theorem~\ref{Lemma imp}, depends on the interplay between $\delta, \lambda$, and the level sets $\gamma_\mathbf{i}, \gamma_\mathbf{u}$. Specifically, selecting a smaller $\lambda$ or operating under a lower disturbance bound (yielding a smaller $\delta$) increases the likelihood of achieving an infinite-horizon guarantee, albeit making condition \eqref{Eq_CBC3} more restrictive.
\end{remark}

While the R-CBC defined in Definition~\ref{R-CBC def} effectively ensures robust safety over both infinite and finite time horizons, its synthesis is computationally infeasible owing to the unknown system dynamics embedded in the left-hand side of \eqref{Eq_CBC3} (\emph{i.e.,} $\mathcal{B}(A \mathcal{M}(x) + B{\mathcal{Q}}(x)u + \omega) $). Although some recent efforts have explored data-driven approaches to address this issue, they often require extensive data over a horizon, which can be costly to acquire. Motivated by this key challenge, we formally define the \emph{physics-guided} data-driven problem that forms the focus of this study.

\begin{resp}
	\begin{problem} \label{P_safety}
		Consider a dt-IANPS in \eqref{Eq_Llike} with unknown matrices $A,B$, and unknown-but-bounded disturbance $\omega$. Develop a \emph{physics-guided} data-driven approach by collecting input-state data from dt-IANPS to design a robust controller that ensures the system's robust safety, while utilizing fundamental physical laws to \emph{mitigate the data} required for safety analysis.
	\end{problem}
\end{resp}

To address Problem \ref{P_safety}, we present our physics-guided data-driven approach for the discrete-time setting in the next section.

\section{Data-Conformity and Physics-Guided Sets}\label{Physics-Guided}

To synthesize an R-CBC for a dt-IANPS under limited data, we first define data-conformity sets that capture consistency with the observed trajectory. We then introduce a physics-guided set that leverages approximate models derived from \emph{fundamental physical laws} to reduce the data requirement. This integration significantly reduces reliance on large datasets while preserving rigorous safety constraints.

\subsection{Data-Conformity (DC) Set}

In our data-driven framework, data is collected from an experiment on \eqref{Eq_Llike} in the presence of \emph{unknown-but-bounded} disturbances. Starting from a given initial condition, we apply a sequence of arbitrary control inputs and record the corresponding state transitions produced by \eqref{Eq_Llike} over time steps $k = 1,2,\dots,T$, with $T \in \mathbb{N}^+$ being the total number of observed samples:
\begin{subequations}\label{Eq_ST}
	\begin{align}
		\Xright &= [ x(1) \quad x(2) \quad\dots\quad x(T)],\label{Eq_STx} \\
		\mathbb{X}&= [ x(0) \quad x(1) \quad\dots \quad x(T-1)],\label{Eq_STa} \\
		\mathbb{U} &= [ u(0) \quad u(1) \quad\dots\quad u(T-1)],\label{Eq_STb} \\
		\mathbb{W} &= [ \omega(0) \quad \omega(1) \quad\dots \quad \omega(T-1)],\label{Eq_STd}
	\end{align}
\end{subequations}
where $\mathbb{W}$ is unknown and cannot be directly measured. Since $\Xright$ and $\mathbb{X}$ are recursively affected by $\mathbb{W}$, it is clear that the data is inherently \emph{noisy}.
Given that the unknown disturbance is bounded, we impose the following bound on the instantaneous \emph{weighted} norm of $\omega$:
\begin{equation}\label{Eq_noise0}
	\Vert\Upsilon\omega\Vert_{{2}} \leq \epsilon_\omega ,
\end{equation}
where $\epsilon_\omega \in \mathbb{R}_{>0}$ is a sufficiently small constant, and $\Upsilon \in \mathbb{R}^{\hat{n} \times n}$ is a full-column-rank weight matrix.

\begin{remark} [\textbf{On Weighted Norm}]\label{Re_weight}
		In general, \eqref{Eq_noise0} represents a \emph{weighted} norm of $\omega$. In the special case where $\Upsilon = \mathds{I}^{n \times n}$, this constraint simplifies to an upper bound on the Euclidean norm of $\omega$. However, in many cases, different components of the system state have varying natures and ranges, making $\Vert\omega\Vert_2$ an inadequate criterion for assessing the components of disturbance. In such cases, by incorporating $\Upsilon$, one obtains linear combinations of the rows of $\omega$, corresponding to different state variables. 
		This enables a more meaningful evaluation of disturbance ranges.
\end{remark}
By defining $\Phi =  \Upsilon^\top \Upsilon$, one can rewrite \eqref{Eq_noise0} as 
\begin{equation}\label{Eq_noise_S}
	\omega^\top \Phi \omega \leq \epsilon_\omega^2.
\end{equation}
Since $\Upsilon$ has full column rank, the matrix $\Phi$ is positive definite and thus invertible. Consequently, the Schur complement \cite{zhang2006schur} can be applied to rewrite \eqref{Eq_noise_S} as
\begin{equation}\label{Eq_noise2}
	\quad \omega\omega^\top  \preceq \epsilon^2_\omega\Phi^{-1}.
\end{equation}
By applying \eqref{Eq_noise2} to the collected data in \eqref{Eq_ST}, one can obtain additional insights into the system matrices $A$ and $B$. In particular, given the availability of an extended dictionary for $\mathcal{M}(x)$ and $\mathcal{Q}(x)$, the following trajectories can be extracted based on \eqref{Eq_STa} and \eqref{Eq_STb}:
\begin{subequations}\label{Eq_mandq}
	\begin{align} 
		\mathbb{M} &\!=\! [ \mathcal{M}(x(0))\,\,\,\mathcal{M}(x(1))\,\,\,\dots\,\,\,\mathcal{M}(x(T\!-\!1))], \\
		{\mathbb{Q}} &\!=\! [ \mathcal{Q}(x(0))u(0)\,\,\,\mathcal{Q}(x(1))u(1) \,\,\,\dots\,\,\,\mathcal{Q}(x(T\!-\!1))u(T\!-\!1))].
	\end{align}
\end{subequations}
Therefore, for $j= 1,\dots,T,$ one has
\begin{align}
	&\Xright_{j}  \!=\! A\mathbb{M}_{j}\!+\!B{\mathbb{Q}}_{j} \!+\!\mathbb{W}_{j}  \!=\! \Omega  \mathbb{Y}_{j} \!+\!\mathbb{W}_{j} , 
	\, 
\end{align}
where  $\Omega = [A\quad B]$, and $\mathbb{Y}_{j}  \!=\! \left[\begin{array}{c} \mathbb{M}_{j} \\
	{\mathbb{Q}}_{j}  \end{array}\right]\!\!. $
	
Accordingly, $\mathbb{W}_{j}   \!=\! \Xright_{j}  \!-\! \Omega  \mathbb{Y}_{j} $. Then, constraint \eqref{Eq_noise2} implies that
\begin{align}
	\epsilon_{\omega}^2\Phi^{-1} \!\succeq& \,\mathbb{W}_{j} \mathbb{W}_{j} ^\top  \nonumber = (\Xright_{ j}  - \Omega  \mathbb{Y}_{ j} )(\Xright_{ j}  - \Omega  \mathbb{Y}_{ j} )^\top   \nonumber\\
	=&\, \Omega  \mathbb{Y}_{ j}  \mathbb{Y}_{ j} ^\top \Omega^\top \!\!-\! \Omega  \mathbb{Y}_{ j} \Xright_{j}^\top \!\!-\!\Xright_{j}  \mathbb{Y}_{j} ^\top  \Omega^\top \!\!+\! \Xright_{j}  \Xright_{j}^\top, \label{Eq_DC}
\end{align}
which yields $T$ matrix inequalities for $\Omega$.

In the following subsection, we introduce a physics-guided set based on approximate models grounded in \emph{fundamental physical principles}, aiming to mitigate the reliance on extensive data.

\subsection{Physics-Guided (PG) Set}

While $A$ and $B$ are unknown in real-world scenarios, the \emph{fundamental physical laws} allow  in many cases for the extraction of \emph{inaccurate yet sufficiently close} nominal matrices, $\tilde{A}$ and $\tilde{B}$, which satisfy the following weighted norm condition:
\begin{gather}\label{Eq_PI0}
	\Vert \Upsilon(\Omega - \tilde\Omega)\Vert_{{2}} \leq \epsilon_\Omega, \\
	\text{with} ~\Omega = [A\quad B], \quad  \tilde\Omega = [ \tilde A\quad  \tilde B],  \nonumber
\end{gather}
for a constant $\epsilon_\Omega \in \mathbb{R}_{>0}$, and the weight matrix $\Upsilon$ as in \eqref{Eq_noise0}. By expanding \eqref{Eq_PI0} and since $\Phi =  \Upsilon^\top \Upsilon$, one has
\begin{gather}\label{Eq_PI_S}
	(\Omega -  \tilde\Omega)^\top \Phi (\Omega - \tilde\Omega) \preceq \epsilon_\Omega^2\mathds{I}_{(m+q)}.
\end{gather}
According to Schur complement \cite{zhang2006schur}, one can reformulate \eqref{Eq_PI_S} as
\begin{align}\label{Eq_PI}
	(\Omega - \tilde\Omega)  (\Omega -  \tilde\Omega)^\top \preceq \epsilon_\Omega^2\Phi^{-1}.
\end{align}
Therefore, alongside the structural physical insights used to derive the extended dictionaries for $\mathcal{M}(x)$ and $\mathcal{Q}(x)$ (as discussed in Remark \ref{Remark: dictionary}), the underlying physical laws provide the crucial parametric constraint in \eqref{Eq_PI}. It is important to note that, in Section~\ref{Sec_Case_studies}, the term ``purely data-driven method'' specifically refers to the synthesis framework without the physics-guided quadratic constraint \eqref{Eq_PI}, while still assuming the availability of the extended dictionaries, consistent with the standard setting in the data-driven literature.

\begin{remark}[\textbf{On Weight $\Upsilon$}]
	~We use the same weight matrix $\Upsilon$ from \eqref{Eq_noise0} in \eqref{Eq_PI0} to linearly transform different rows of $\Omega-\tilde{\Omega}$ associated with different state components. As discussed in Remark \ref{Re_weight}, $\Upsilon$ serves to normalize these components by homogenizing their ranges. Therefore, using identical weights in both \eqref{Eq_noise0} and \eqref{Eq_PI0} is more appropriate.
\end{remark}

\subsection{Extraction of DC and PG Sets}\label{Subsec_PIex}

After introducing the data-driven and physics-guided sets, we now turn to the practical question of how these sets can be extracted. For instance, in the case of a rotating rigid spacecraft, considered as our second case study, the discretized Euler equations are derived from fundamental physical principles, including Newton’s second law for rotational motion and the conservation of angular momentum, and are given by:

\begin{align}\label{Spacecraft}
	\tilde{\Sigma}\!:\!\begin{cases}
		x_1^+ = x_1 + \frac{0.02}{J_1} ((J_2-J_3)x_2x_3 + u_1), \\
		x_2^+ = x_2 + \frac{0.02}{J_2} ((J_3-J_1)x_3x_1 + u_2), \\
		x_3^+ = x_3 + \frac{0.02}{J_3} ((J_1-J_2)x_1x_2 + u_3),
	\end{cases}
\end{align}
where $x^+ \coloneq x(k + 1), \; k \in \mathbb{N}$. In addition, $x_1$ to $x_3$ represent the angular velocity components along the principal axes, $u_1$ to $u_3$ are the torque inputs, and $J_1$ to $J_3$ are the principal moments of inertia. Given this model, the degree of monomials in $\mathcal{M}(x)$ and $\mathcal{Q}(x)$ could be considered $2$ and $0$, respectively. However, one might consider more conservative upper bounds, such as $3$ and $1$, to ensure that all nonlinear terms in the actual system are accounted for in our analysis.

Furthermore, although the physical model yields approximate matrices $\tilde{A}$ and $\tilde{B}$, the actual system dynamics may deviate due to various sources of uncertainties including measurement inaccuracies in $J_1$ to $J_3$. Based on prior knowledge, an upper bound $\epsilon_\Omega$ on the spectral norm of the difference between the true matrices $A, B$ and their nominal counterparts can also be estimated. Additionally, to account for uncertainties arising from modeling inaccuracies, we include a disturbance term $\omega$ as in \eqref{Eq_Llike}, where the upper bound in \eqref{Eq_noise0} limits this source of uncertainty.

Having introduced the data-conformity and  physics-guided sets, we now proceed to propose our physics-guided data-driven framework for discrete-time systems in the following section.

\section{Physics-Guided Data-Driven Design of R-CBC and R-SC}\label{R-CBC_Design}
This section details the core methodology for jointly synthesizing the R-CBC and its R-SC. We establish the theoretical conditions, outline the algorithmic procedure, and analyze both computational complexity and problem feasibility.
\subsection{Synthesis Theorem and Algorithmic Design}
Here, we first specify our R-CBC and its controller as
\begin{align}\label{controller}
	\mathcal{B}(x) = x^\top  P x, \quad u = K(x)x, 
\end{align}
where $P\succ 0$. Note that $K(x)$ is not restricted to the same monomials as the system dynamics and may contain all the monomials up to a certain degree. By doing so, one can simplify the closed-loop form of system \eqref{Eq_Llike} as follows:
\begin{align}\notag
	x^+ =\,\,\,\,\,\,&A \mathcal{M}(x) + B{\mathcal{Q}}(x)u + \omega\\\notag
	\stackrel{\eqref{transform}, \eqref{controller}}{=} \, & \, (A \mathcal{C}(x) + B{\mathcal{Q}}(x)K(x))x + \omega \\\label{Eq_CLs}
	=\,\,\,\,\,&  \Omega Z(x)x + \omega, ~ \text{with} ~ Z(x) =  \left[\begin{array}{c} \mathcal{C}(x) \\
		\mathcal{Q}(x)K(x)\end{array}\right]\!\!.
\end{align}

Before presenting the main result of this section, we first introduce the following lemma, which is essential for showing the subsequent theorem.

\begin{lemma}[\textbf{Bounding $\omega$}]\label{Lemma 1}
	Consider a positive-definite matrix $P$ and 
	positive constants $\mu, \delta \in \mathbb{R}_{>0}$. Assuming that $\omega$ satisfies \eqref{Eq_noise2}, one has
	\begin{equation}\label{Eq_delta_cons}
		P^{-1} \succeq \delta^{-1}(1+\mu^{-1})\epsilon_\omega^2 \Phi^{-1} {\Rightarrow}  (1+\mu^{-1})\omega^\top P \omega \preceq \delta.
	\end{equation}
\end{lemma}\vspace{0.2cm}

\begin{proof}
	Let $P^{-1} \succeq \delta^{-1}(1+\mu^{-1})\epsilon_\omega^2 \Phi^{-1}$. By applying \eqref{Eq_noise2}, one has
	\begin{align}\label{New98}
		P^{-1} \succeq \delta^{-1}(1+\mu^{-1})\omega\omega^\top\!. 
	\end{align}
	By employing the Schur complement \cite{zhang2006schur}, inequality \eqref{New98} can be reformulated as
	\begin{align}\label{Eq_noise_limit}
		(1+\mu^{-1})\omega^\top P \omega \preceq \delta,
	\end{align}
	which concludes the proof.
\end{proof}

By leveraging Lemma~\ref{Lemma 1}, which establishes a valid upper bound for the disturbance in terms of $\delta$, we are now ready to present the main result of this section, derived via the S-procedure and subsequently relaxed into an SOS optimization program using Positivstellensatz in Subsection~\ref{Subsec: SOS_Relax}.

\begin{theorem}[\textbf{R-CBC and R-SC Design for dt-IANPS}] \label{T_final}
	Consider a dt-IANPS $\Sigma$ as in Definition \ref{Def_1}. Let there exist $\bar\gamma_{\mathbf i}, \bar\gamma_{\mathbf u}  \in \mathbb{R}_{>0}$, with $\bar\gamma_\mathbf i > \bar\gamma_\mathbf u$, $\delta \in \mathbb{R}_{\geq 0}, \lambda \in (0,1)$, matrix $\bar P\succ 0$,  polynomial matrix $\bar K(x)$, and   $\kappa_{j=0,\dots,T} \!\!:\mathbb{R}^n \to  \mathbb{R}_{\geq0}$, such that
	\begin{subequations}\label{Eq_CBC_theorem}
		\begin{align}\label{Eq_noise_theorem}
			&\bar P - \bar\delta (1+\mu^{-1})\epsilon_\omega^2\Phi^{-1} \succeq 0, \\\label{Eq_CBC1_theorem}
			&\bar P - \bar\gamma_{\mathbf i}\nu_\mathbf i\nu_\mathbf i^\top ~\!\succeq 0, \\\label{Eq_CBC2_theorem}
			-&\bar P  + \bar\gamma_{\mathbf u}\nu_\mathbf u\nu_\mathbf u^\top  \succeq   0,\\\notag
			& 
			\begin{bmatrix} \lambda \bar P && 0  && 0\\ * &&  0 && \left[\begin{array}{c} -\mathcal{C}(x)\bar P \\
					-\mathcal{Q}(x)\bar K(x)\end{array}\right] \\
				* &&  * &&  (1+\mu)^{-1} \bar P\end{bmatrix} + \kappa_0(x)\begin{bmatrix}\mathcal N^{PG} & 0  \\ * &  0 \\
			\end{bmatrix} \\\label{Eq_CBC3_theorem} 
			& ~~+ \sum_{j=1}^{T}\kappa_j(x)\begin{bmatrix}\mathcal N^{DC}_{j} & 0  \\ * &  0 \\
			\end{bmatrix}  \succeq 0,   \quad \quad \forall x \in \tilde X,
		\end{align}
	\end{subequations}
	where
	\begin{subequations}
		\begin{align}\label{Eq_CBC1_nuc}
			X_\mathbf i &\subseteq \{x\in \mathbb{R}^n\!\!: xx^\top \preceq\nu_\mathbf i\nu_\mathbf i^\top\, , \nu_\mathbf i \in \mathbb{R}^{{n \times n}} \},\\ \label{Eq_CBC2_nuc}
			X_\mathbf u &\subseteq \{x\in \mathbb{R}^n\!\!: xx^\top\succeq\nu_\mathbf u\nu_\mathbf u^\top\, , \nu_\mathbf u \in \mathbb{R}^{{n \times n}} \},\\\label{Eq_PI_new2}
			\mathcal N^{PG} &=\begin{bmatrix} \tilde\Omega \tilde\Omega^\top  -\epsilon_\Omega^2\Phi^{-1}  & - \tilde\Omega\\
				* & \mathds{I}\end{bmatrix}\!\!,\\\label{Eq_PI_new1}
			\mathcal N^{DC}_{j} &=\begin{bmatrix}
				\Xright_{j}  \Xright_{j}^\top - \epsilon_{\omega}^2\Phi^{-1} &  -\Xright_{j} \mathbb{Y}_{j} ^\top\\
				* & \mathbb{Y}_{j}  \mathbb{Y}_{j} ^\top\end{bmatrix}\!\!,
		\end{align}
	\end{subequations}
	for  some $\mu \in \mathbb{R}_{>0}$. Then, $\mathcal{B}(x) = x^\top P x$, with $P = \bar P^{-1}$, is an R-CBC  for the  dt-IANPS and $u = K(x)x$, with $K(x) =  {\bar K}(x) \bar P^{-1} =  {\bar K}(x)P$, is its corresponding R-SC, with $\gamma_\mathbf i  = \bar\gamma_\mathbf i ^{-1}$, $\gamma_\mathbf u = \bar\gamma_\mathbf u^{-1}$ (where $\gamma_\mathbf i < \gamma_\mathbf u$), and $\delta = \bar\delta^{-1}$.
\end{theorem}

\begin{proof}
	We first show that condition \eqref{Eq_CBC3_theorem}  ensures the satisfaction of condition \eqref{Eq_CBC3}. Since $\mathcal{B}(x) = x^\top P x$ and by defining $ \Bxright := \mathcal{B}(A\mathcal{M}(x)+B{\mathcal{Q}}(x)u + \omega),$ we have
	\begin{align} \notag
		\Bxright &\stackrel{\eqref{Eq_CLs}}{=}\,(\Omega Z(x)x + \omega)^\top P(\Omega Z(x)x + \omega)\\\label{Eq_barriernext} 
		& \,\,= J(x)^\top P J(x)+ 2\underbrace{J(x)^\top \sqrt{P}}_{a^\top}\underbrace{\sqrt{P}\omega}_{b} + \omega^\top P \omega,  
	\end{align}
	where  $J(x)= \Omega Z(x)x$. According to the Cauchy-Schwarz inequality, \emph{i.e.,}  $a b \leq \Vert a \Vert_2 \Vert b \Vert_2,$ for any $a^\top, b \in \R^{n}$, followed by
	employing Young's inequality \cite{young1912classes}, \emph{i.e.,} $\Vert a \Vert_2 \Vert b \Vert_2 \leq \frac{\mu}{2} \Vert a \Vert^2_2 + \frac{1}{2\mu} \Vert b \Vert^2_2$, for any $\mu \in \mathbb{R}_{>0}$, one has
	\begin{align} \label{Eq_nextBB}
		\Bxright \leq & (1+\mu)J(x)^\top P J(x) + (1+\mu^{-1})\omega^\top P \omega. 
	\end{align}
	Given the satisfaction of \eqref{Eq_noise_theorem} and according to Lemma \ref{Lemma 1}, one can bound the disturbance in \eqref{Eq_nextBB}, resulting in
	\begin{align} \notag
		\Bxright - \lambda\mathcal{B}(x) \leq  (1+\mu)J(x)^\top P J(x) - \lambda \overbrace{x^\top P x}^{\mathcal{B}(x)}+ \delta.
	\end{align}
	By defining 
	\begin{align}\notag  
		G(x) := (1+\mu)J(x)^\top P J(x) - \lambda x^\top P x,
	\end{align}
	it is clear that if $G(x)\leq 0$, then 
	\begin{align} \notag
		\Bxright - \lambda\mathcal{B}(x) \leq \delta.
	\end{align}
	Therefore, we now focus on satisfying the constraint $G(x)\leq 0$. One can expand $G(x)$ as
	\begin{align*}
		G(x) =\,& (1+\mu)J(x)^\top P J(x) - \lambda x^\top P x \notag \\
		=\,& (1+\mu)(\Omega Z(x)x)^\top P\Omega Z(x)x - \lambda x^\top P x \notag \\
		=\,& x^\top \big((1+\mu)Z(x)^\top \Omega ^\top P\Omega Z(x)-\lambda P\big) x. 
	\end{align*}
	To enforce $G(x)\leq 0$, it is sufficient to satisfy
	\begin{align}\label{new09}
		(1+\mu)Z(x)^\top \Omega ^\top P\Omega Z(x)- \lambda P \preceq 0. 
	\end{align}
	By Schur complement~\cite{zhang2006schur}, this inequality is equivalent to
	\begin{align} 
		(1+\mu)\Omega Z(x) P^{-1}Z(x)^\top \Omega ^\top- \lambda P^{-1} \preceq 0, \notag
	\end{align}
	which could be rewritten in the following quadratic form:
	\begin{align}\label{Eq_CBC3_new}
		\mathcal{H}^{CBC}(\Omega, x):=\thetaTildeI^\top \mathcal N^{CBC}(x) \thetaTildeI \preceq 0,
	\end{align}
	with
	\begin{multline*}
		\mathcal N^{CBC}(x) = \\ \!\!\begin{bmatrix} -\lambda P^{-1} & \!\!0  \notag\\ * & \!\!(1+\mu) \underbrace{\left[\begin{array}{c} \mathcal{C}(x) \notag\\
					\mathcal{Q}(x)K(x)\end{array}\right]}_{Z(x)} P^{-1}\underbrace{\left[\begin{array}{c} \mathcal{C}(x) \notag\\
					\mathcal{Q}(x)K(x)\end{array}\right]^\top}_{Z(x)^\top} \end{bmatrix}\!\!.
	\end{multline*}
	There are two challenges in satisfying \eqref{Eq_CBC3_new}: first, the exact value of the matrix $\Omega$ is unknown, and second, the design variables $P^{-1}$ and $K(x)$ (in $Z(x)P^{-1}Z(x)^\top$) are bilinear. Let us discuss the first challenge and address the second one at a later stage. Although the exact value of $\Omega$ is unknown, our data-conformity and physics-guided constraints provide $T + 1$ quadratic matrix inequalities involving $\Omega$ (\emph{i.e.,} \eqref{Eq_DC} and \eqref{Eq_PI}). Equation \eqref{Eq_DC} can be reformulated as
	\begin{gather} \label{Eq_DC_new}
		\mathcal{H}^{DC}_j(\Omega) :=\thetaTildeI^\top \mathcal N^{DC}_{j} \thetaTildeI \preceq 0,
	\end{gather}
	with $\mathcal N^{DC}_{j}$ as in \eqref{Eq_PI_new1}. Similarly, \eqref{Eq_PI} can be rewritten as
	\begin{align} \label{Eq_PI_new}
		\mathcal{H}^{PG}(\Omega):=\thetaTildeI^\top \mathcal N^{PG} \thetaTildeI \preceq 0, 
	\end{align}
	with $\mathcal N^{PG}$ as in \eqref{Eq_PI_new2}.
	By applying S-procedure \cite{9308978}, to enforce \eqref{Eq_CBC3_new} where \eqref{Eq_DC_new} and \eqref{Eq_PI_new} are fulfilled, it is sufficient to show that there exists $\kappa_{j=0,\dots,T}(x):\mathbb{R}^n \to  \mathbb{R}_{\geq0}$ such that
	\begin{align} \label{Eq_Slemma} 
		& \mathcal N^{CBC}(x)- \kappa_0(x)\mathcal N^{PG}- \sum_{j=1}^{T}\kappa_j(x)\mathcal N^{DC}_j \preceq 0.
	\end{align}
	Let us now address the bilinearity between matrices $P^{-1}$ and $K(x)$ in $Z(x)P^{-1}Z(x)^\top$, as the second issue. To do so, one can use dilation \cite{caverly2019lmi} and show that inequality \eqref{Eq_Slemma} is equivalent to \eqref{Eq_CBC3_theorem}, \emph{i.e.,}
	\begin{align*}
		\eqref{Eq_Slemma}\Leftrightarrow   \eqref{Eq_CBC3_theorem}.
	\end{align*}
	Given the satisfaction of  \eqref{Eq_CBC3_theorem}, it is clear that condition \eqref{Eq_Slemma} is also met. By fixing suitable scalar values for $\mu > 0$ and $\lambda\in (0,1)$, \eqref{Eq_CBC3_theorem} is a linear matrix inequality (LMI) based on the design variables $\bar P, \bar K(x),$ and $\kappa_{j=0,\dots,T}(x)$. 
	
	\noindent 
	As the final step of the proof, we show that satisfying conditions \eqref{Eq_CBC1_theorem} and \eqref{Eq_CBC2_theorem} implies the fulfillment of conditions \eqref{Eq_CBC1} and \eqref{Eq_CBC2}, respectively.
	Since $P\succ 0$ and $\gamma_\mathbf i, \gamma_\mathbf u \in \mathbb{R}_{>0}, $ by applying Schur complement, one can verify that:
	\begin{align*}
		\gamma_\mathbf i - x^\top Px \geq  0~&\Leftrightarrow ~ P^{-1} - \gamma_\mathbf i^{-1}xx^\top \succeq 0,
	\end{align*}
	which is equivalent to \eqref{Eq_CBC1}. Similarly, according to Schur complement, one has 
	\begin{align*}
		\gamma_\mathbf u - x^\top Px >  0~&\Leftrightarrow ~ P^{-1} - \gamma_\mathbf u^{-1}xx^\top \succ 0,
	\end{align*}
	implying that their complements are also equivalent:
	\begin{align}\label{New76}
		\gamma_\mathbf u - x^\top Px \ngtr  0~&\Leftrightarrow ~ P^{-1} - \gamma_\mathbf u^{-1}xx^\top \nsucc 0.
	\end{align}
	It is clear that the left-hand side of \eqref{New76} is equivalent to \eqref{Eq_CBC2}. Therefore, we need to satisfy $ P^{-1} - \gamma_\mathbf u^{-1}xx^\top \nsucc 0$.  Since $\nsucc$ for matrices implies that the matrix is either negative semi-definite or indefinite, and the latter case is challenging to enforce, a conservative condition with ``$\preceq$'' can be used instead. This, combined with \eqref{Eq_CBC1_nuc} and \eqref{Eq_CBC2_nuc}, guarantees that satisfying conditions \eqref{Eq_CBC1_theorem} and \eqref{Eq_CBC2_theorem} implies the fulfillment of conditions \eqref{Eq_CBC1} and \eqref{Eq_CBC2}, which completes the proof.
\end{proof}

\begin{remark}[\textbf{Purely Data-Driven Method}]\label{Re_PDD}
Our framework accommodates scenarios where physical priors are either unavailable or uninformative. If no prior physical information ($\tilde{A}, \tilde{B}$) is available, the physics-guided constraint \eqref{Eq_PI_new} is fundamentally absent, and the term $\kappa_0(x)\begin{bmatrix}\mathcal N^{PG} & 0 \\ * & 0 \end{bmatrix}$ could be explicitly removed from \eqref{Eq_CBC3_theorem} as discussed in Section \ref{Sec_Case_studies} (cf. the provided comparisons in Table~\ref{Tab_results}). Alternatively, if a nominal model is available but its associated error bound $\epsilon_\Omega$ is excessively conservative, the term remains in the formulation, although the PG constraint no longer effectively restricts the data-conformity set. Consequently, the SOS solver naturally drives the multiplier $\kappa_0(x)$ toward zero, practically discarding the constraint without causing infeasibility. Both scenarios degrade the optimization to a purely data-driven approach, inherently requiring more data samples to achieve the same safety guarantee.
\end{remark}

\begin{remark}[\textbf{On Offline Computation of $\nu_\mathbf i, \nu_\mathbf u$}]\label{Re_nunu}
~The bounding matrices $\nu_\mathbf{i}, \nu_\mathbf{u} \in \mathbb{R}^{n \times n}$ in \eqref{Eq_CBC1_nuc} and \eqref{Eq_CBC2_nuc} are parameters derived from the safety specifications prior to solving the SOS optimization program. Assuming $X_\mathbf i$ and $X_\mathbf u$ can be conservatively bounded by spheres of radii $r_\mathbf i, r_\mathbf u \in \mathbb{R}_{>0}$, such that $\Vert x\Vert_2 \leq r_\mathbf i$ for all $x \in X_\mathbf i$ and $\Vert x\Vert_2 \geq r_\mathbf u$ for all $x \in X_\mathbf u$, a straightforward choice is
\begin{align}\label{New54}
	\nu_\mathbf i\nu_\mathbf i^\top = r_\mathbf i^2\mathds{I}_n, \quad\nu_\mathbf u\nu_\mathbf u^\top= r_\mathbf u^2\mathds{I}_n.
\end{align}
For complex geometries, less conservative matrices can be computed offline via standard convex optimization by finding the minimum volume enclosing ellipsoid for $X_\mathbf{i}$ and the maximum-volume ellipsoid strictly disjoint from $X_\mathbf{u}$.
\end{remark}

\begin{remark}[\textbf{On Selecting \(\mu\) and \(\lambda\)}]\label{Re_New}
    ~The  parameter \(\mu\) originates from the Young's inequality \cite{young1912classes}: \[
    2J(x)^\top P\,\omega \leq \mu\,J(x)^\top P\,J(x) + \mu^{-1}\,\omega^\top P\,\omega.
    \]
    Since \(J(x)=x^+-\omega\) and typically \(\Vert x^+\Vert_2\gg \Vert \omega\Vert_2\), one can fix \(\mu\) with a sufficiently small value to balance the terms \(\mu\,J(x)^\top P\,J(x)\) and \(\mu^{-1}\,\omega^\top P\,\omega\), leading to a tighter bound. In addition, the parameter \(\lambda\in(0,1)\) can be initialized with a small value (\emph{e.g.,} 0.1) and incrementally increased using a fixed step size until a valid solution is found.
\end{remark}

With the theoretical conditions established, we now turn to the computational implementation.

\subsection{SOS Relaxation and Tractability} \label{Subsec: SOS_Relax}
While conditions \eqref{Eq_noise_theorem}--\eqref{Eq_CBC2_theorem} are standard Linear Matrix Inequalities (LMIs), the condition in \eqref{Eq_CBC3_theorem} is a polynomial matrix inequality. To render this condition computationally tractable for Algorithm~\ref{alg}, we employ a sum-of-squares (SOS) relaxation via the Positivstellensatz.

Assuming the state space $X$ is a basic semi-algebraic set defined by a vector of known polynomial inequalities $X = \{x \in \mathbb{R}^n \mid \alpha(x) \ge \boldsymbol{0}\}$, we introduce a vector of SOS scalar polynomial multipliers $\Gamma(x)$. We can then relax \eqref{Eq_CBC3_theorem} into the following strictly tractable SOS matrix condition:
\begin{align} \notag
	&\begin{bmatrix} \lambda \bar P & 0  & 0\\ * &  0 & \left[\begin{array}{c} -\mathcal{C}(x)\bar P \\
			-\mathcal{Q}(x)\bar K(x)\end{array}\right] \\
		* &  * &  (1+\mu)^{-1} \bar P\end{bmatrix} + \kappa_0(x)\begin{bmatrix}\mathcal N^{PG} & 0  \\ * &  0 \\
	\end{bmatrix} \\\label{Eq_SOSLag}
	&\,\,\,\,\,\,\,\,\,\,\,\,\,+ \sum_{j=1}^{T}\kappa_j(x)\begin{bmatrix}\mathcal N^{DC}_{j} & 0  \\ * &  0 \\
	\end{bmatrix} - \alpha^\top(x) \Gamma(x)\mathds{I}_{2n+m+q} \in \mathcal{K}_{sos}, 
\end{align}
where $\mathcal{K}_{sos}$ denotes the cone of SOS matrices. Since the components of $\Gamma(x)$ are strictly SOS and $\alpha(x) \ge \boldsymbol{0}$ for all $x \in X$, the subtracted scalar term $\alpha^\top(x)\Gamma(x)$ is guaranteed to be non-negative over $X$. This strictly enforces the original condition \eqref{Eq_CBC3_theorem} over the bounded operating domain, mitigating numerical ill-conditioning that arises from enforcing SOS constraints globally over the Euclidean state space.

We introduce Algorithm \ref{alg}, which outlines the required steps in Theorem \ref{T_final} and its SOS relaxation for the physics-guided data-driven design of R-CBC and its R-SC in the discrete-time setting.

\subsection{Computational Complexity Analysis}
   Here, we briefly discuss the computational complexity of our proposed framework. Since Algorithm~\ref{alg} relies on a sum-of-squares (SOS) optimization program, its scalability is primarily governed by the state-space dimension $n$ and the maximum degree $h$ of the system's polynomial dynamics. Our method introduces decision variables including the $n\times n$ matrix $\bar P=P^{-1}$, the $l\times n$ polynomial matrix $\bar K(x)$ (of degree $\bar h$), $T+1$ scalar polynomial functions $\kappa_j(x)$ (of degree $\hat h$), and the multiplier vector $\Gamma(x)$ (of degree $\tilde h$) corresponding to the domain constraints. The number of coefficients for these polynomials scales combinatorially, given by $\binom{n+\bar h}{\bar h}$, $\binom{n+\hat h}{\hat h}$, and $\binom{n+\tilde h}{\tilde h}$, respectively. Since these polynomial degrees are typically chosen proportional to $h$, the number of coefficients becomes the decisive factor in computational complexity. Consequently, the computational burden grows \emph{polynomially} with $n$ (for a fixed $h$) and with $h$ (for a fixed $n$). Nevertheless, as demonstrated in our simulation results, our method can efficiently manage systems with relatively complex dynamics.

\subsection{Feasibility Analysis}\label{Sec-Feasibility}

Our approach, consistent with existing literature on model-based Lyapunov and barrier functions, only provides \emph{sufficient} conditions for ensuring the safety of nonlinear polynomial systems. Here, we present a feasibility analysis to provide insight into the situations under which a solution is more likely to exist.

In our problem formulation, \eqref{Eq_noise_theorem}--\eqref{Eq_CBC2_theorem} involve $\bar P = P^{-1}$ and scalar variables $\bar \gamma_{\mathbf{i}} =\gamma_{\mathbf{i}}^{-1}$, $\bar \gamma_{\mathbf{u}} =\gamma_{\mathbf{u}}^{-1}$, and $\bar \delta = \delta^{-1}$. To interpret these conditions, consider a simpler case where $\nu_\mathbf{i}\nu_\mathbf{i}^\top = r_\mathbf{i}^2\mathds{I}_n$ and $\nu_\mathbf{u}\nu_\mathbf{u}^\top = r_\mathbf{u}^2\mathds{I}_n$, with $r_\mathbf{i} < r_\mathbf{u}$ being two scalar values. From \eqref{Eq_CBC1_theorem} and \eqref{Eq_CBC2_theorem}, one can derive $\bar\gamma_\mathbf{i} \leq \lambda_{min}(\bar P)/r_\mathbf{i}^2$ and $\bar\gamma_{\mathbf{u}} \geq \lambda_{max}(\bar P)/r_\mathbf{u}^2$. For $\gamma_{\mathbf{i}} < \gamma_{\mathbf{u}}$ (equivalently $\bar\gamma_{\mathbf{i}} > \bar\gamma_{\mathbf{u}}$) to hold, the condition number of $\bar P = P^{-1}$ (\emph{i.e.,} $\frac{\lambda_{\text{max}}(\bar P)}{\lambda_{\text{min}}(\bar P)}$) should be small and close to 1. Similarly, considering conditions \eqref{Eq_noise_theorem} and \eqref{Eq_CBC2_theorem}, a smaller condition number allows a larger ratio of $\dfrac{\gamma_{\mathbf{u}}}{\delta}$, which is desirable according to~\eqref{New9}. Thus, conditions \eqref{Eq_noise_theorem}--\eqref{Eq_CBC2_theorem} primarily constrain certain characteristics of $\bar P = P^{-1}$.

	\begin{algorithm}[t!] 
	\caption{Physics-guided data-driven design of R-CBC and R-SC for dt-IANPS}\label{alg}
	\begin{algorithmic}[1]
		\REQUIRE The state set $X$, bounds for initial and unsafe sets $\nu_\mathbf i, \nu_\mathbf u$ as in \eqref{New54} as part of the safety specification, extended dictionaries $\mathcal{M}(x), \mathcal{Q}(x)$, and $\epsilon_\omega, \epsilon_\Omega, \Upsilon$ as in \eqref{Eq_noise0}  and \eqref{Eq_PI0}
		\STATE Collect $\Xright,\mathbb{X},  \mathbb{U}$ as in~\eqref{Eq_ST}
		\STATE Form $\mathbb{M}, {\mathbb{Q}}$ as in \eqref{Eq_mandq}, and $\Phi$ as $\Upsilon^\top \Upsilon$
		\STATE Initialize $\lambda \in (0,1)$, $\mu \in \mathbb{R}_{>0}$ according to Remark~\ref{Re_New}
		\STATE Find $\bar P$, $\bar K(x)$, $\bar\gamma_\mathbf i$, $\bar\gamma_\mathbf u$ (with $\bar\gamma_\mathbf i >\bar\gamma_\mathbf u$), and $\bar\delta$ that satisfy~\eqref{Eq_noise_theorem}, \eqref{Eq_CBC1_theorem}, \eqref{Eq_CBC2_theorem}  and~\eqref{Eq_SOSLag} using SeDuMi and SOSTOOLS~\cite{prajna2004sostools}
		\STATE Construct $\mathcal{B}(x)=x^\top P x$ using $P = \bar{P}^{-1}$, and $u = K(x)x$, with $K(x) =  {\bar K}(x) \bar P^{-1} =  {\bar K}(x)P$
		\STATE Construct $\gamma_\mathbf i  = \bar\gamma_\mathbf i ^{-1}$, $\gamma_\mathbf u = \bar\gamma_\mathbf u^{-1}$ (where $\gamma_\mathbf i < \gamma_\mathbf u$), and $\delta = \bar\delta^{-1}$
		\STATE Given designed parameters $\lambda, \gamma_\mathbf i, \gamma_\mathbf u$, and $\delta$, check conditions \eqref{New9}, \eqref{new8} and provide safety guarantee  for either \emph{infinite or finite} time horizons
		\ENSURE  R-CBC~$\mathcal{B}(x) = x^\top P x$, R-SC $u\!=\! K(x) x$, and guaranteed robust safety for unknown dt-IANPS
	\end{algorithmic}
\end{algorithm}

\begin{table*}[t!]  
	\makeatletter
	\long\def\@makecaption#1#2{%
		\vskip\abovecaptionskip
		\noindent \textbf{#1.} #2\par
		\vskip\belowcaptionskip}
	\makeatother
	\centering
	\caption{A comparison of the sample sizes required to guarantee safety over an infinite time horizon using our physics-guided method ($T_{\text{PGDD}}$) versus a purely data-driven approach ($T_{\text{DD}}$), where inequality \eqref{Eq_PI} is unavailable.  The reported values for $\gamma_{\mathbf{i}}$, $\gamma_{\mathbf{u}}$, $\delta$, and runtime (\texttt{RT}) correspond to the results from our physics-guided method. For all experiments, $\Phi = \mathds{I}_3$.}\vspace{0.2cm}
	\label{Tab_results}
	{\small
		\begin{tabular}{@{}lcccccccccccc@{}}
			\toprule
			System & $\epsilon_\omega$  & $\epsilon_\Omega$  & $\lambda$ & $\gamma_\mathbf{i}$ & $\gamma_\mathbf{u}$ & $\delta$ & $T_{PGDD}$ & $T_{DD} $ & $\mathcal{T}$ & \texttt{RT} (sec)  \\
			\midrule
			{Lorenz} & $0.001$ &  $0.1$ & $0.99$ & $6.71 \times 10^6$  &  $1.19 \times 10^7$ & $2.79 \times 10^3$ & $2$  &   $13$ &  $\infty$  & $2.21$ \\
			\midrule
			{Spacecraft} & $0.05$ & $0.8$  & $0.99$  & $7.19 \times 10^5$  &  $9.46 \times 10^5$ & $6.02 \times 10^3$  & $15$  &   $31$ &  $\infty$ & $3.43$ \\
			\midrule
			{Higher-Degree}  & $0.0014$ & $0.325$  & $0.99$  & $1.40 \times 10^7$  &  $1.83 \times 10^7$ & $8.95 \times 10^3$ & $13$  &   $35$ &  $\infty$  & $30.35$ \\
			\bottomrule
		\end{tabular}
	}
\end{table*}

Furthermore, restricting the barrier certificate to a quadratic form ($\mathcal{B}(x) = x^\top P x$) inherently constrains the level sets to ellipsoids. While this is a trade-off for SOS tractability and data efficiency, it can induce geometric conservatism; separating highly non-convex initial and unsafe sets with a single ellipsoid may be restrictive and impact LMI feasibility.

The most key constraint, however, is condition \eqref{Eq_CBC3_theorem}. Since the exact matrix $\Omega$ is unknown, safety should be enforced uniformly across all models satisfying the data and physical constraints \eqref{Eq_DC} and \eqref{Eq_PI0}. Consequently, feasibility hinges on sufficiently shrinking this admissible set, since an excessively large set may contain pathological models (\emph{e.g.,} uncontrollable or inherently unsafe dynamics) while simultaneously requiring a single pair of R-CBC and R-SC to remain valid across all models within the set. The shrinkage of the admissible set occurs when (i) the physical prior is highly accurate (yielding a tight $\epsilon_\Omega$), or (ii) the trajectory data is sufficiently informative and long, mitigating the need for strong priors. Conversely, overly conservative bounds ($\epsilon_\omega, \epsilon_\Omega$) or poor data excitation unnecessarily inflate the uncertainty set, potentially risking SOS infeasibility. Finally, the introduction of the multiplier $\Gamma(x)$ in the SOS relaxation \eqref{Eq_SOSLag} further enhances feasibility by confining the strict positivity requirement to the operating domain $X$, thereby avoiding the substantial conservatism associated with enforcing it globally.

\section{Simulation Results}\label{Sec_Case_studies}

This section evaluates the efficacy of our physics-guided data-driven framework through three benchmark case studies: a Lorenz system~\cite{lopez2019synchronization}, serving as a classical chaotic nonlinear polynomial model; a rotating rigid spacecraft~\cite{khalil2002nonlinear}; and a system featuring higher-degree polynomial dynamics. As summarized in Table~\ref{Tab_results}, we systematically compare our approach against its purely data-driven counterpart (cf. Remark~\ref{Re_PDD}) to explicitly demonstrate how integrating prior physical knowledge enables formal safety guarantees using significantly shorter trajectories. The Lorenz system, in particular, is selected because it effectively captures complex, chaotic dynamics prevalent across various scientific domains. Such chaotic models are integral to secure communications for signal encryption~\cite{wang2009chaotic}, atmospheric modeling in weather prediction~\cite{Lorenz1963}, adaptive robotics in unpredictable environments~\cite{sprott2010elegant}, and neuroscience for simulating brain activity to understand disorders like epilepsy~\cite{strogatz2024nonlinear}. All simulations were executed in MATLAB on a macOS device equipped with an M3 Max chip. Furthermore, in all case studies, $\nu_\mathbf{i}$ and $\nu_\mathbf{u}$ are derived based on the spherical assumption (cf. Remark \ref{Re_nunu}).

\subsection{Lorenz System} 
To illustrate the applicability of our approach, we begin with the discrete-time Lorenz system, a well-known nonlinear polynomial chaotic system. The nominal model dynamics are described by
\begin{equation}\label{Eq_Lorenz_nominal}
	\tilde{\Sigma}\!:\!\begin{cases}
		x_1^+=x_1 + 0.02(10x_2 -10x_1),  \\
		x_2^+=x_2 + 0.02(28x_1 - x_2- x_1 x_2+u),\\
		x_3^+=x_3+ 0.02(x_1x_3 - \frac{8}{3}x_3).
	\end{cases}
\end{equation}

Based on this model, we construct the extended dictionary $\mathcal{M}(x)$ by including all monomials of state variables up to degree $2$,  while $\mathcal{Q}(x)$ is considered to be constant due to the presence of a single input component, \emph{i.e.,}
\begin{align}\notag
	\mathcal{M}(x) \!=\! \begin{bmatrix}
		x_1;\! 
		x_2;\! 
		x_3;\! 
		x_1x_2;\! 
		x_2x_3;\! 
		x_1x_3;\! 
		x_1^2;\! 
		x_2^2;\! 
		x_3^2
	\end{bmatrix}\!,~
	\mathcal{Q}(x) \! =\!  1.
\end{align}
From \eqref{Eq_Lorenz_nominal}, we extract the nominal matrices  
	$\tilde{A}$ and  
	$\tilde{B}$, which allows us to rewrite the nominal model as
	\begin{align*}
		\tilde{\Sigma}\!:  x^+ = \tilde{A} \mathcal{M}(x) + \tilde{B}{\mathcal{Q}}(x)u.
	\end{align*}
However, as discussed earlier, this nominal model does not fully capture the true system dynamics. The actual behavior is influenced by the unknown matrix $\Omega$ and additive disturbance $\omega$, leading to the accurate unknown dynamics
\begin{align} \label{Eq_real_behaviour}
	\Sigma\!:\! x^+ = \Omega \begin{bmatrix}
		\mathcal{M}(x) \\
		\mathcal{Q}(x)u
	\end{bmatrix} + \omega.
\end{align}
In our simulation, the unknown matrix $\Omega$ is generated by perturbing each element of the nominal matrix $\tilde\Omega = [ \tilde A\quad  \tilde B]$ with a random number drawn from the interval $[-0.0025, 0.0025]$. Additionally, each component of the disturbance vector $\omega$ is randomly sampled at every time step from $[-0.004, 0.004]$.
The chosen values for $\epsilon_\Omega$ and $\epsilon_\omega$ are sufficiently large to accommodate these perturbations.

The regions of interest are given as $X = [-15,15]^3$, $X_\mathbf i = [0,2] \times [-2,2]^2$, and $X_\mathbf u = ([-15,-6]^2 \times [6,15]) \cup ([-15, 15]\times[10, 15]\times[-15, 15])$. Within this setup, we restrict the polynomial degree of $\bar{K}(x)$ to $1$ (resulting in a control input $u$ of degree 2), and all coefficients $\kappa_j(x)$ to $2$. We also set $\mu = 0.002$ and execute Algorithm~\ref{alg} to synthesize the R-CBC and its robust controller.

We design the matrix $P$ and controller $u$ in our physics-guided setting as 
\begin{align}\notag
	P &= 10^5 \times \begin{bmatrix}
		2.554 & 0.612 & -1.255 \\
		0.612 & 1.615 & -1.140 \\
		-1.255 & -1.140 & 4.565 \\
	\end{bmatrix}\!\!,\\\notag
	u &=   0.189x_1^2 + 1.014x_1x_2 + 1.417x_1x_3 + 0.334x_2^2 \\\notag
	&~~~+ 0.220
	x_2x_3 + 0.052x_3^2 - 46.329x_1 - 52.205x_2 \\\label{kdj}
	&~~~+ 25.766x_3.
\end{align}	

As shown in Table \ref{Tab_results}, incorporating the physics-guided quadratic constraint allowed us to achieve an \emph{infinite} time horizon safety guarantee with only $2$ data samples (\emph{i.e.,} $T = 2$). In contrast, the purely data-driven case required at least \emph{$13$} data samples ($T = 13$) to ensure the same guarantee.

With the robust safety controller in place, all trajectories of the Lorenz system remain within the safe set for an infinite time horizon, as shown in Figure \ref{fig:b1}, aligning with our theoretical results in Theorem \ref{Lemma imp}.

\subsection{Rotating Rigid Spacecraft}
As our second case study, we investigate the dynamics of a rotating rigid spacecraft as presented in~\eqref{Spacecraft}, where $J_1 = 0.5$, $J_2 = 1$, and $J_3 = 1.3$.
The system dictionary $\mathcal{M}(x)$ includes monomials up to degree 2, and $\mathcal{Q}(x)$ is considered constant:
\begin{align}\notag
	\mathcal{M}(x) \!=\! \begin{bmatrix}
		x_1;\! 
		x_2;\! 
		x_3;\! 
		x_1x_2;\! 
		x_2x_3;\! 
		x_1x_3;\! 
		x_1^2;\! 
		x_2^2;\! 
		x_3^2
	\end{bmatrix}\!,\quad 
	\mathcal{Q}(x) \! =\!  \mathds{I}_{3}.
\end{align}
Due to various sources of uncertainty, the physical model is inaccurate and the real behavior of the system is determined by the unknown matrix $\Omega$ and the disturbance $\omega$ as in \eqref{Eq_real_behaviour}.  The accurate matrix $\Omega$ is simulated by adding random perturbations within $[-0.001, 0.001]$ to each entry of the nominal matrix $\tilde{\Omega}$. Additionally, each component of $\omega$ is sampled uniformly from $[-0.02, 0.02]$ at each time step. The amounts for $\epsilon_\Omega$ and $\epsilon_\omega$ are chosen large enough to capture these perturbations. 

The sets of interest are given as $X = [-25,25]^3$, $X_\mathbf i = [-5,5]^3$, and $X_\mathbf u = ([-25, -15]\times[0, 25]\times[-25, 25]) \cup [10,25]^3  \cup ([10, 25]\times[-25, -10]^2)$. We set $\mu=0.004$, the maximum degree of $\bar{K}(x)$ to $1$, and the maximum degree of each $\kappa_j(x)$ to $2$. 
The matrix $P$  and controller components $u_i$ in our physics-guided scheme are designed as
\begin{align}\notag
	P &= 10^3 \times \begin{bmatrix}
		4.923 & -0.046 & 0.006 \\
		-0.046 & 9.305 & 1.204 \\
		0.006 & 1.204 & 4.487 \\
	\end{bmatrix}\!\!,\\\notag
	u_1 &= 0.537x_1^2 + 0.775x_1x_2 + 0.814x_1x_3 + 0.785x_2^2 \\\notag
	&~~~+ 0.486x_2x_3 - 0.759x_3^2 - 17.478x_1 + 0.048x_2 
	\\\notag&~~~- 0.030x_3,\\\notag
	u_2 &= 0.157x_1^2 + 1.042x_1x_2 - 0.772x_1x_3 + 0.850x_2^2 \\\notag
	&~~~+ 4.251
	x_2x_3 + 1.097x_3^2 - 1.318x_1 - 38.637x_2 \\\notag
	&~~~+ 1.435x_3,\\\notag
	u_3 &=   -0.188 x_1^2 + 0.332x_1x_2 - 0.253x_1x_3 + 0.730x_2^2 
	\\\notag&~~~+ 1.273x_2x_3 + 4.770x_3^2 - 0.476x_1 + 1.299x_2 \\\label{newjh}&~~~- 45.065x_3.
\end{align}

As shown in Table~\ref{Tab_results}, incorporating the physics-guided quadratic constraint enables us to guarantee safety over an infinite time horizon using only \emph{$T = 15$} data samples, whereas the purely data-driven approach requires at least \emph{$T = 31$} samples to achieve the same level of guarantee. This highlights the sample efficiency of our physics-guided framework. It is also worth noting that the increased data requirement in this case, compared to the Lorenz benchmark in Table~\ref{Tab_results}, stems from the relatively large values chosen for $\epsilon_\Omega$ and $\epsilon_\omega$, which demand more data to ensure the same level of safety. Figure~\ref{fig:b2} illustrates that the designed controller successfully maintains all trajectories within the safe set.

\subsection{Higher-Degree Polynomial System} \label{Example_HD}
To further assess the capabilities of our method in the discrete-time setting, we extend the spacecraft model in \eqref{Spacecraft} to include degree-3 polynomial dynamics (\emph{i.e.,} $x_1^3, x_2x_3^2, x_1x_2x_3$). To do so, we augment the nominal model for rigid spacecraft body with polynomial terms of degree three,
with $\mathcal{M}(x)$ containing all the $19$ monomials of the components of $x$ up to the degree $3$. The accurate model used to generate samples is simulated by adding random perturbations from $[-0.002,0.002]$ to the entries of the nominal matrix $\tilde{\Omega}$. Additionally, each entry of $\omega$ is sampled uniformly from $[-0.0002, 0.0002]$ at each time step.

\begin{figure}[t]
	\centering
	\includegraphics[width= 0.6\linewidth]{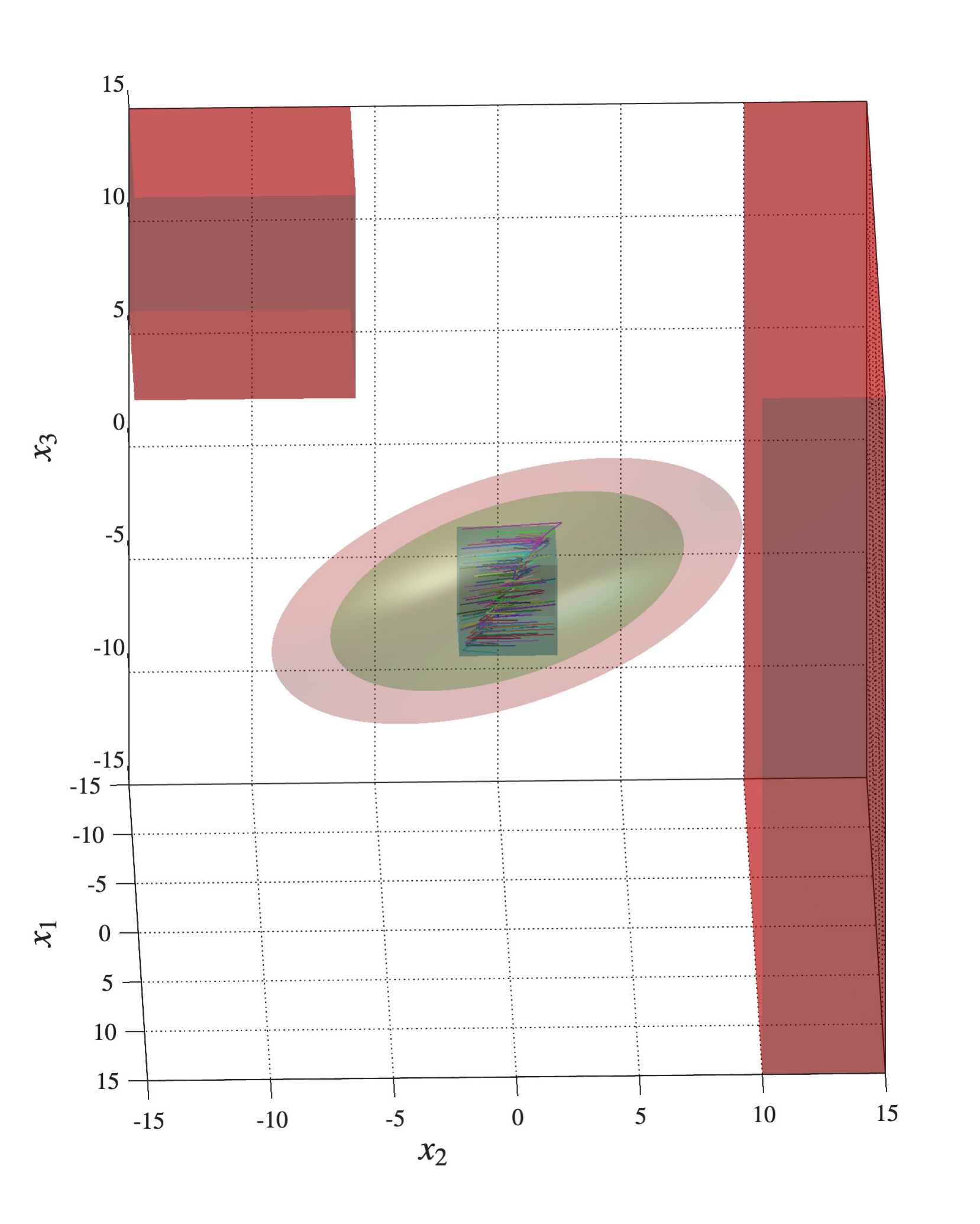}
	\caption{{\bf Lorenz system:} Closed-loop state trajectories of the Lorenz system under the designed controller~\eqref{kdj}, starting from different initial states in $X_\mathbf i \in [0,2] \times [-2, 2]^2$. Initial and unsafe regions are depicted by green \protect\greensquare\ and red \protect\redsquare\ boxes, respectively. The boundaries $\mathcal{B}(x) = \gamma_{\mathbf{i}}$ and $\mathcal{B}(x) = \gamma_{\mathbf{u}}$ are indicated by green and red ellipsoids, respectively. The simulations are generated with $200$ different initial states and disturbances satisfying \eqref{Eq_noise0}, demonstrating the robustness of our framework to disturbances.}
	\label{fig:b1}
\end{figure}

\begin{figure}[t]
	\centering
	\makebox[\linewidth]{\includegraphics[width=0.7\linewidth]{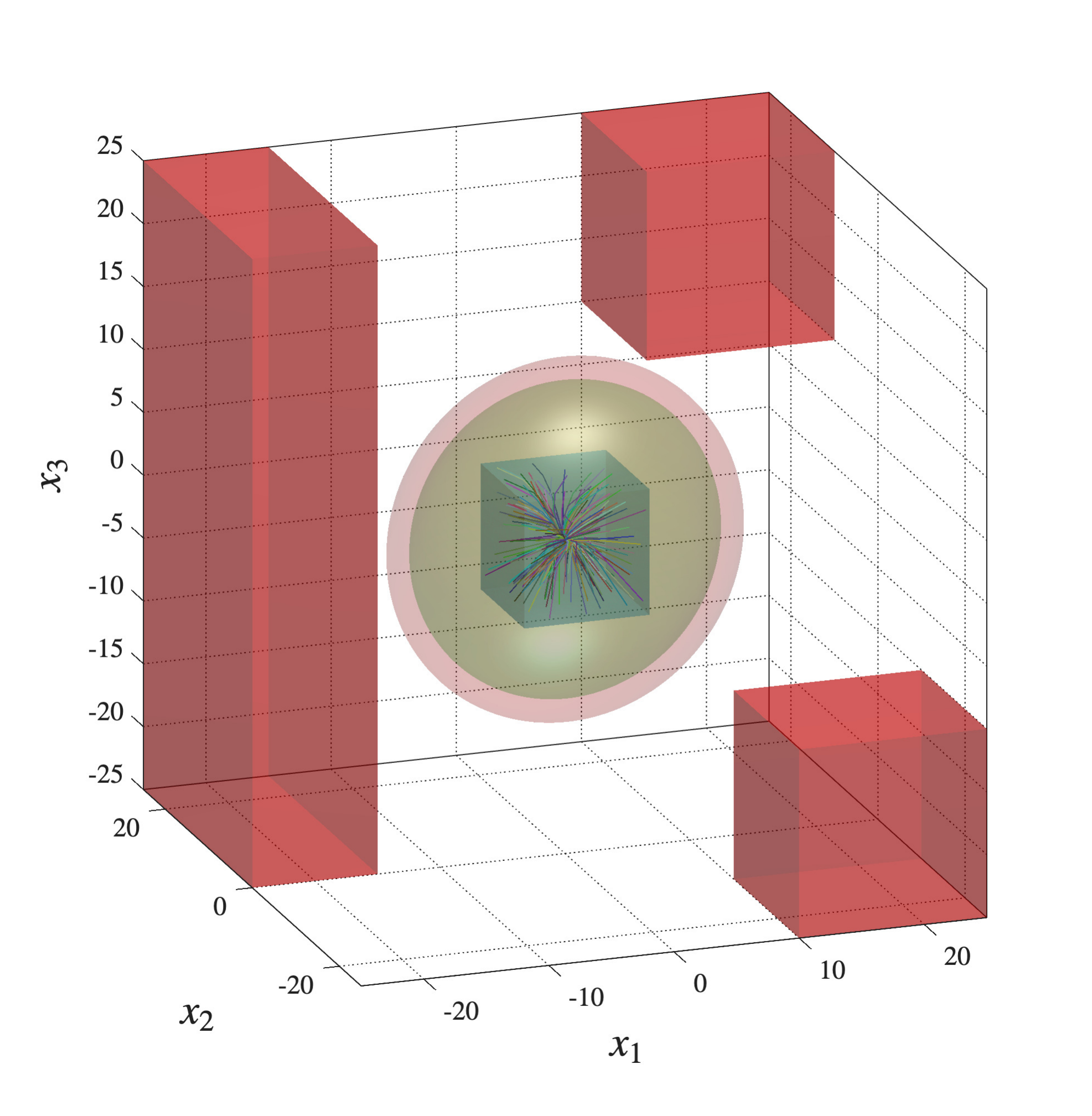}}\vspace{-0.45cm}
	\caption{{\bf Spacecraft system:} Closed-loop state trajectories of the spacecraft system under the designed controller in~\eqref{newjh}. The simulations are performed using $200$ distinct initial conditions, highlighting our framework’s robustness to uncertainty. Initial and unsafe zones are depicted by green \protect\greensquare\ and red \protect\redsquare\ boxes, respectively, while the boundary $\mathcal{B}(x) = \gamma_{\mathbf{i}}$ and $\mathcal{B}(x) = \gamma_{\mathbf{u}}$ are depicted as green and red ellipsoids.}
	\label{fig:b2}
\end{figure}

Regions of interest are given as $X = [-25,25]^3$, $X_\mathbf i = [-5,5]^3$, and $X_\mathbf u = [-25, -12]^3 \cup [12, 25]^3$.
We set the maximum degree for $\bar{K}(x)$ and all $\kappa_j(x)$ to 2, and choose $\mu = 4 \times 10^{-5}$. We design the matrix $P$ in our physics-guided setting as
$$P = 10^4 \times \begin{bmatrix}
	5.841 & -3.140 & 2.205 \\
	-3.140 & 1.286 & -6.098 \\
	2.205 & -6.098 & 9.790 \\
\end{bmatrix}\!\!.
$$
The designed controller components are not reported due to their large size.

With the physics-guided quadratic constraint, an infinite-horizon safety guarantee is achieved using only $T = 13$ data samples. In contrast, the purely data-driven approach requires at least $T = 35$ samples to achieve the same level of safety. Notably, even though degree-3 monomials were included into the system dynamics, resulting in increased nonlinearity, the number of required samples remained comparable to the previous spacecraft case study. This is primarily due to the lower disturbance levels considered here, which reduce uncertainty and allow for more efficient use of data in the safety assurance process. Figure~\ref{fig:b3} shows that the designed controller effectively keeps all trajectories within the safe set.

\begin{figure}[t]
	\centering
	\makebox[\linewidth]{\includegraphics[width=0.75\linewidth]{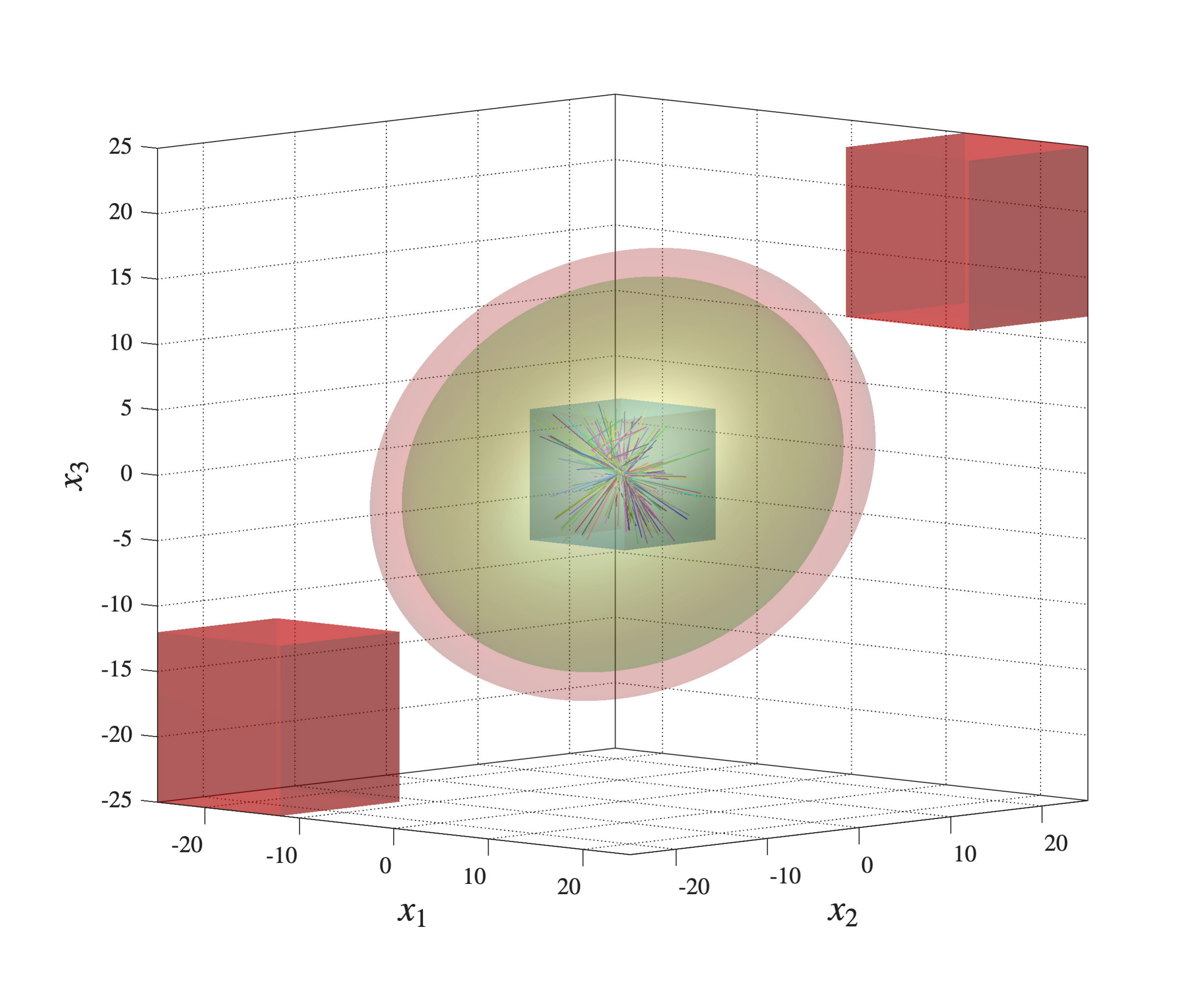}}
	\caption{{\bf Higher-degree polynomial system:} Closed-loop state trajectories of the higher-degree polynomial system under the designed controller. Simulations are performed using $200$ distinct initial conditions, illustrating the framework’s robustness to uncertainties. Initial and unsafe regions are depicted with green \protect\greensquare\ and red \protect\redsquare\ boxes, respectively. The boundaries $\mathcal{B}(x) = \gamma_{\mathbf{i}}$ and $\mathcal{B}(x) = \gamma_{\mathbf{u}}$ are visualized as green and red ellipsoids.}
	\label{fig:b3}
\end{figure}

\section{Conclusion}\label{Sec_Conclusion}

We developed a \emph{physics-guided} data-driven framework for synthesizing robust safety controllers for discrete-time nonlinear polynomial systems under unknown-but-bounded disturbances. Our approach utilized a single input-state trajectory to construct robust control barrier certificates despite \emph{noisy data}, ensuring safety guarantees even with model uncertainty. Unlike conventional trajectory-based methods that require long horizons for safety analysis, our proposed scheme incorporated fundamental physical principles, reducing data dependency with a \emph{shorter} trajectory, while preserving robustness. To achieve this, the proposed synthesis method was cast as an SOS optimization problem that jointly constructs the R-CBC and the corresponding R-SC by integrating observed data with approximate physical models. The effectiveness of the framework was validated across three benchmark examples, demonstrating its capability to ensure robust safety with \emph{significantly lower} data requirements. Promising directions for future work include incorporating \emph{compositional} techniques \cite{akbarzadeh2025formal} to address the scalability limitations of monolithic SOS programs for high-dimensional systems by decomposing the global synthesis problem into smaller, tractable subproblems. Furthermore, extending the framework to \emph{stochastic} control barrier certificates for systems affected by heavy-tailed noise, as well as investigating more expressive composite barrier functions to reduce geometric conservatism, constitute promising directions for future research.

\bibliographystyle{IEEEtran}
\bibliography{biblio}	

\newpage 

\begin{IEEEbiography}[{\includegraphics[width=1in,height=1.3in,clip,keepaspectratio]{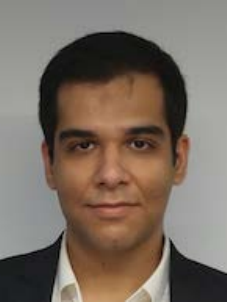}}]{MohammadHossein Ashoori}~(Student Member, IEEE) received his B.Sc. and M.Sc. degrees in electrical engineering from Sharif University of Technology (SUT), Tehran, Iran, in 2019 and 2022, respectively. He is currently pursuing his PhD in
the School of Computing at Newcastle University,
UK.  His research
interests include cyber-physical systems (CPS), computer vision, and digital signal processing.
\end{IEEEbiography}\vspace{-2.7cm}

\begin{IEEEbiography}[{\includegraphics[width=1in,height=1.3in,clip,keepaspectratio]{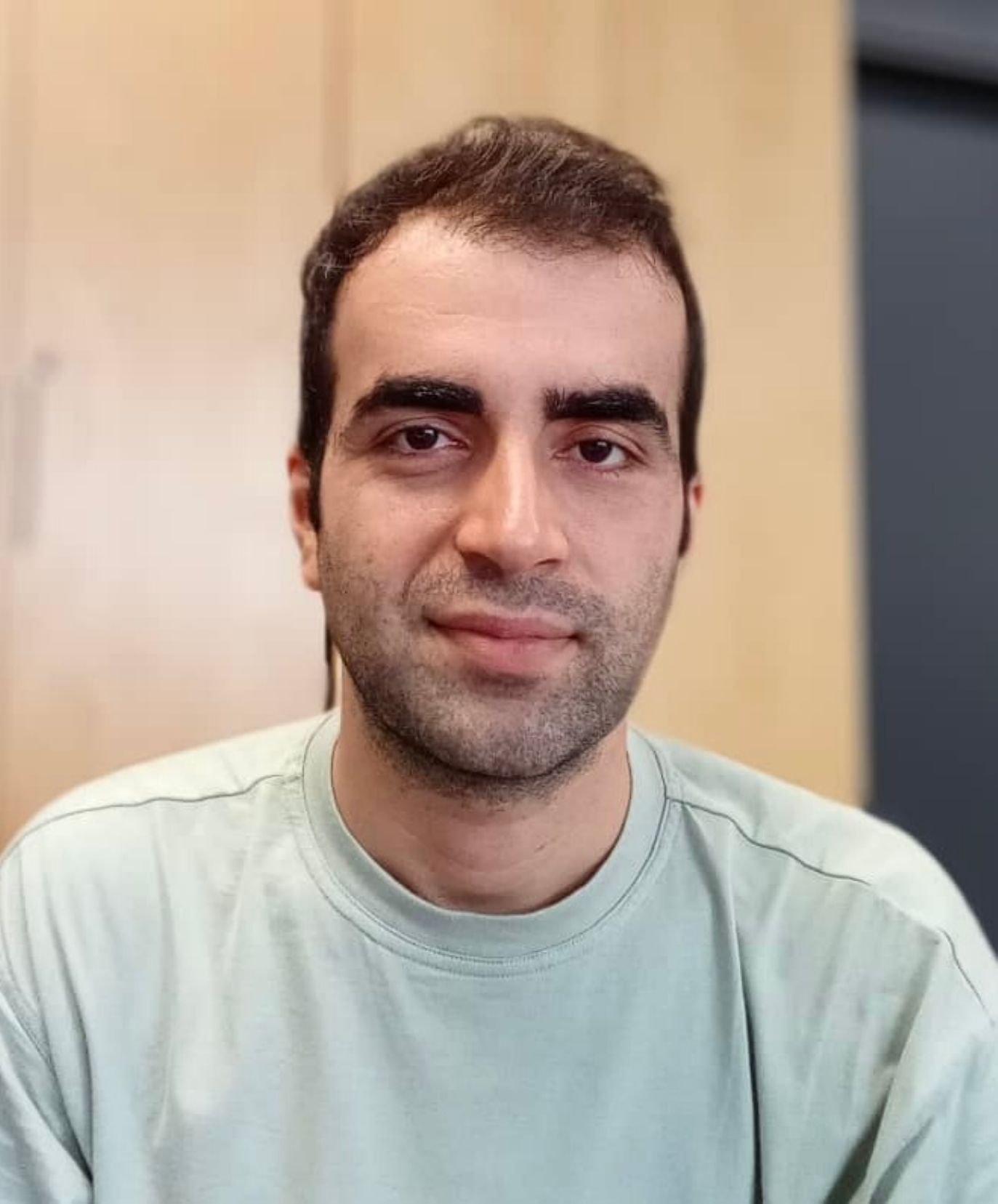}}]{Ali Aminzadeh}~(Member, IEEE) is a Postdoctoral Researcher in the Automation Technology and Mechanical Engineering Unit at the Faculty of Engineering and Natural Sciences, Tampere University, Finland. He received his Ph.D. in Aerospace Engineering from K. N. Toosi University of Technology, Tehran, Iran, in 2023, specializing in Flight Dynamics and Control. He completed his M.Sc. in the same field from the University of Tehran (UT), Iran, in 2016. His research interests focus on cooperative control of multi-agent systems, collective decision-making, data-driven control, and formal verification techniques.
\end{IEEEbiography}\vspace{-2.7cm}

\begin{IEEEbiography}[{\includegraphics[width=1in,height=1.25in,clip,keepaspectratio]{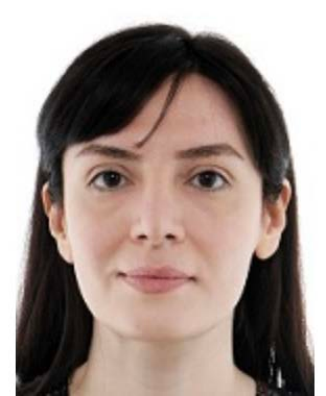}}]{Amy Nejati}~(M'18--SM'25) is an Assistant Professor in the School of Computing at Newcastle University in the United Kingdom. Prior to this, she was a Postdoctoral Associate at the Max Planck Institute for Software Systems in Germany from July 2023 to May 2024. She also served as a Senior Researcher in the Computer Science Department at the Ludwig Maximilian University of Munich (LMU) from November 2022 to June 2023. She received the PhD in Electrical Engineering from the Technical University of Munich (TUM) in 2023. She has received the B.Sc. and M.Sc. degrees both in Electrical Engineering. Her line of research mainly focuses on developing efficient (data-driven) techniques to design and control highly-reliable autonomous systems while providing mathematical guarantees. She was named a CPS Rising Star 2024, and her paper was selected as a Best Repeatability Prize Finalist at ACM HSCC 2025.\end{IEEEbiography}\vspace{-2.7cm}

\begin{IEEEbiography}[{\includegraphics[width=1in,height=1.25in,clip,keepaspectratio]{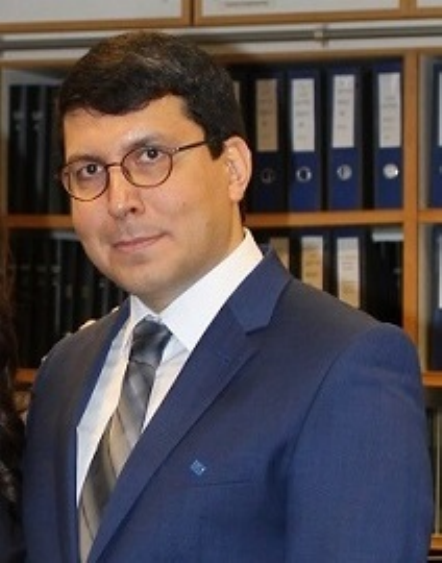}}]{Abolfazl Lavaei}~(M'17--SM'22) is an Assistant Professor in the School of Computing at Newcastle University, United Kingdom. Between January 2021 and July 2022, he was a Postdoctoral Associate in the Institute for Dynamic Systems and Control at ETH Zurich, Switzerland. He was also a Postdoctoral Researcher in the Department of Computer Science at LMU Munich, Germany, between November 2019 and January 2021. He received the Ph.D. degree in Electrical Engineering from the Technical University of Munich (TUM), Germany, in 2019. He obtained the M.Sc. degree in Aerospace Engineering with specialization in Flight Dynamics and Control from the University of Tehran (UT), Iran, in 2014. He is the recipient of several international awards in the acknowledgment of his work including  Best Repeatability Prize (Finalist) at the ACM HSCC 2025, IFAC ADHS 2024, and IFAC ADHS 2021, HSCC Best Demo/Poster Awards 2022 and 2020, IFAC Young Author Award Finalist 2019, and Best Graduate Student Award 2014 at University of Tehran with the full GPA (20/20). His line of research primarily focuses on the intersection of Control Theory, Formal Methods, and Statistical Learning Theory.
\end{IEEEbiography}

\end{document}